\documentclass{article}
\usepackage{lmodern}
\usepackage{graphicx} 
\usepackage[backend=biber, style=alphabetic, maxnames=99]{biblatex}
\usepackage[margin=1in]{geometry}
\usepackage{todonotes}
\usepackage{amsmath}
\usepackage{hyperref}
\usepackage[capitalize]{cleveref} 
\usepackage{amssymb}
\usepackage{amsthm}
\usepackage{mathtools}
\usepackage{stmaryrd}
\usepackage{appendix}
\usepackage{algorithm}
\usepackage{algpseudocode}
\usepackage{empheq}
\usepackage{authblk}

\title{On Extensions of the Unanimous Vote Problem}
\author[1]{Evan J. R. Brody}
\author[1]{Haya Diwan}
\author[1]{Lisa Hellerstein}
\author[2]{Thomas Lidbetter}
\affil[1]{New York University\thanks{\texttt{\{evan.brody,hd2371,lisa.hellerstein\}@nyu.edu}}}
\affil[2]{Rutgers Business School\thanks{\texttt{tlidbetter@business.rutgers.edu}}}

\date{}

\newtheorem{theorem}{Theorem}
\newtheorem{definition}{Definition}

\newtheorem{corollary}{Corollary}
\newtheorem{proposition}{Proposition}

\newtheorem{lemma}{Lemma}
\newtheorem{fact}{Fact}
\crefname{fact}{Fact}{Facts}

\newcommand{\expect}{\operatorname{\mathbb{E}}\expectarg}
\DeclarePairedDelimiterX{\expectarg}[1]{[}{]}{%
  \ifnum\currentgrouptype=16 \else\begingroup\fi
  \activatebar#1
  \ifnum\currentgrouptype=16 \else\endgroup\fi
}

\newcommand{\innermid}{\nonscript\;\delimsize\vert\nonscript\;}
\newcommand{\activatebar}{%
  \begingroup\lccode`\~=`\|
  \lowercase{\endgroup\let~}\innermid 
  \mathcode`|=\string"8000
}

\newcommand{\highlight}[2][yellow]{\mathchoice%
  {\colorbox{#1}{$\displaystyle#2$}}%
  {\colorbox{#1}{$\textstyle#2$}}%
  {\colorbox{#1}{$\scriptstyle#2$}}%
  {\colorbox{#1}{$\scriptscriptstyle#2$}}}%

\newcommand{\nmin}{{N_{1}}}
\newcommand{\nmax}{{N_{n}}}
\newcommand{\eqdef}{\stackrel{\textnormal{def}}{=}}

\DeclareMathOperator{\cost}{cost}
\DeclareMathOperator*{\argmin}{argmin}
\DeclareMathOperator*{\argmax}{argmax}
\DeclareMathOperator{\OPT}{OPT}

\newcommand{\MulAdGapLink}[1]{
\href{https://drive.google.com/file/d/14h_SgAhZYv4fg5zoJdACpdyGdxAqyIeO/view?usp=sharing}{#1}
}
\newcommand{\AddAdGapLink}[1]{
\href{https://drive.google.com/file/d/1A_O8tUlD9Wntyw3GfF5Q9-d1AD9dEHn0/view?usp=sharing}{#1}
}
\newcommand{\DesmosLink}[1]{
\href{https://www.desmos.com/3d/rnniay2kya}{#1}
}

\begin{document}

\maketitle

\vspace{50pt}
\begin{abstract}
    The Unanimous Vote problem is to determine a fixed order in which to flip each of $n$ biased coins, where each coin can be flipped only once, such that the expected number of flips until seeing both a head and a tail (or flipping all coins) is minimized.
    \cite{dumankeles2026} gave an $\mathcal{O}(n \log n)$-time algorithm for this problem. Extensions of the Unanimous Vote problem are a rich source of stochastic optimization problems.
    We focus on three: (1) a variant in which each coin can be flipped arbitrarily many times (a solution is thus an infinite sequence of coin choices), (2) a generalization with $d$-sided dice, that can each be rolled once, where dice must be rolled until two different outcomes are observed (or all dice have been rolled), and (3) a different generalization with $d$-sided dice, where dice must be rolled until all $d$ outcomes have been observed. For (1), we show that there is an optimal sequence which follows a simple greedy rule; the same rule only gives a 1-additive approximation for the original problem~\cite{dumankeles2026}. The rule also yields a correspondence between a particular optimal sequence and a related mechanical word, which we exploit to characterize the conditions under which this optimal sequence is periodic. We establish tight multiplicative and additive adaptivity gaps for this variant. For (2), we show that two different generalizations of the greedy rule from \cite{dumankeles2026} can be combined to obtain a PTAS. For (3), we give an $\mathcal{O}(\log d)$-approximation algorithm by reducing the problem to Submodular Ranking~\cite{azar2011}; the same reduction technique can be used to yield approximation algorithms for other stochastic probing problems. Finally, we pose a number of related open questions.
\end{abstract}

\newpage

\section{Introduction}

We are given $n$ biased, independent coins with biases \(p_1\leq\ldots\leq p_n\), where $p_i$ is the probability of heads for coin $i$. Each coin can be flipped only once. We are asked to arrange the coins in some order, and to flip the coins in that order until we have observed both a head and a tail or have flipped all coins. The Unanimous Vote problem asks: what ordering minimizes the expected number of flips? The problem was introduced in \cite{gkenosis2018}, and can be equivalently viewed as asking how to quickly determine whether a set of \(n\) voters have voted unanimously, where the \(i\)\textsuperscript{th} voter votes YES with probability \(p_i\). \cite{gkenosis2018} gave a \(\varphi\)-approximation algorithm, where \(\varphi\approx1.618\) is the golden ratio; \cite{dumankeles2026} gave an exact \(\mathcal{O}(n\log n)\)-time algorithm, resolving the complexity of the problem. The Unanimous Vote problem falls under the broader category of Stochastic Boolean Function Evaluation (SBFE) and is notable as the first nontrivial SBFE problem whose \textit{non-adaptive} variant is known to be solvable in polynomial time.

This paper focuses on three extensions of the (non-adaptive) Unanimous Vote problem. In the first, the Unlimited-Flips variant, we allow for unlimited flips of each of the $n$ coins.  An optimal solution to this variant is an infinite sequence (instead of a permutation) that minimizes the expected number of flips. Optimal solutions in this variant have more elegant structural properties than in the original problem. We show that the (multiplicative) adaptivity gap for the Unlimited-Flips variant is exactly \(1.2\).  This quantity is the maximum ratio of the optimal non-adaptive cost to the optimal adaptive cost; in \cite{dumankeles2026}, the adaptivity gap for the original Unanimous Vote problem is shown to be \(1.2\pm\mathcal{O}\left(\frac{1}{2^n}\right)\), using techniques that are substantially more involved than ours. The additive adaptivity gap, the maximum difference between the optimal non-adaptive cost and the optimal adaptive cost, was implicitly narrowed to be in \(\left[\frac{1}{2}, 1\right]\) by \cite{dumankeles2026}. We prove a tight bound of \(\frac{1}{2}\) for the Unlimited-Flips variant. We also consider the periodicity of optimal sequences for the Unlimited-Flips variant, and show a connection to mechanical words. We apply known results on mechanical words to show that a particular optimal sequence is periodic if and only if \(\log_{\frac{p_n}{1-p_n}}\left(\frac{1-p_1}{p_1}\right)\) is rational.

The second variant we consider is a natural generalization of the Unanimous Vote problem that replaces each coin with a die with \(d\) sides numbered 1 through $d$, where each die can be rolled at most once. We stop rolling dice when we have seen two distinct outcomes (or all dice have been rolled). Equivalently, viewing the outcome of each die roll as a vote for one of $d$ options, we stop rolling dice when we can determine whether the $n$ votes are unanimous. An optimal solution minimizes the expected number of rolls. We show that the 1-additive approximate greedy algorithm included in \cite{dumankeles2026} can be generalized in an unexpected way to give the same approximation for this \(d\)-ary problem. We also extend the ideas from this algorithm to give a polynomial-time approximation scheme (PTAS). If \(d\) is part of the input, solving the problem exactly is \(\mathsf{NP}\)-hard (see Appendix \ref{app: d-ary np-hard}).

For the third variant, we consider the generalization in which each coin is again replaced by a biased die with \(d\) sides numbered 1 through $d$, and each die can be rolled at most once, but we now stop rolling only when all \(d\) distinct outcomes have been observed (or all dice have been rolled). We give an \(\mathcal{O}(\log d)\)-approximation algorithm for this problem, which relies on a reduction to the problem of Submodular Ranking~\cite{azar2011} through the construction of monotone submodular functions. We believe the approach of this reduction is of independent interest as it can be applied more broadly to other non-adaptive stochastic probing problems. We end the paper by posing a number of related open questions.

\subsection{Related Work}

Most of the problems considered in this paper are Stochastic Function Evaluation (SFE) problems. SBFE problems are a well-studied subclass of these problems.  For a broader view of these problems, see the recent review of Ünlüyurt~\cite{unluyurt2025}.

The Unanimous Vote problem was first introduced as a special case of \textit{Stochastic Score Classification}: an SFE problem concerned with determining to which of \(B\) ``classes'' the value of \(\sum_{i=1}^{n}X_iw_i\) belongs, where the \(X_i\) are the probed (queried) Bernoulli random variables and \(w_i \geq 0\) is the weight of \(X_i\)~\cite{gkenosis2018}. For the non-adaptive case with arbitrary probe costs, \cite{ghuge2021ssclass} gives an \(\mathcal{O}(1)\)-approximation to the optimal adaptive algorithm. For unit weights and arbitrary probe costs, \cite{plank2024} gives a \(3+2\sqrt{2}\approx 5.828\)-approximation. In the case of unit weights and unit probe costs, \cite{nielsen2025} gives a PTAS.\par

Other stochastic probing problems related to vote-counting were considered in \cite{hellerstein2024}, which gives several \(\mathcal{O}(1)\)-approximation algorithms. There, the goal was to minimize the expected cost of determining which of $d$ candidates had been elected, in an election with $n$ voters.

The reduction to Submodular Ranking in \Cref{sec: fccp} can be applied more generally, and can be seen as a non-adaptive analogue of the \textit{$Q$-value} approach of \cite{deshpande2016,bach2018}.  That approach is based on constructing a (realization-dependent) monotone submodular ``goal'' function.  We discuss this in more detail in \Cref{sec: reduction generalization}.

Ghuge et al.~\cite{ghuge2026} recently introduced the matroid Stochastic Boolean Function Certification problem, a stochastic probing problem with Boolean-valued probes on correlated distributions. Using LP-based techniques, they give an $\mathcal{O}(\log n)$-approximation algorithm for the general problem and a constant-factor approximation for the special case of uniform matroids. Our approach of reducing to Submodular Ranking yields an alternative $\mathcal{O}(\log n)$-approximation algorithm for the general problem, but it cannot be used to achieve a better approximation bound for uniform matroids.

\section{The Unlimited-Flips Unanimous Vote Problem}\label{sec: unlimited flips}

\subsection{Preliminaries}

We will use the notion of the \textit{bias} of a position from \cite{dumankeles2026}. A position \(k\) in a permutation or sequence \(\sigma\) is said to be \textit{0-biased} (resp. \textit{1-biased}) if it is strictly more probable that the flips of positions \(1,\ldots,k-1\) in \(\sigma\) are identically tails (resp. heads) than identically heads (resp. tails). A position is said to be \textit{unbiased} if these probabilities are equal. For integers \(a,b\) with \(a\leq b\), we use the notation \(\llbracket a, b\rrbracket\eqdef [a,b]\cap\mathbb{Z}\) and for positive integers \(n\) we use \([n]\eqdef\{1,\ldots,n\}\).

\subsection{Greedy is Optimal}

The algorithm of \cite{dumankeles2026} shows that an optimal permutation follows a simple greedy rule at every position except one. The greedy rule chooses a coin with minimal \(p_i\) 
if the position is 1-biased, and a coin with maximal \(p_i\) otherwise; the exceptional position is the one containing the last coin in the permutation that changes the bias. Roughly speaking, the greedy rule breaks because of properties at the end of the permutation.
This raises the question of whether the exceptional position would still exist if the flipping sequence were infinite. We address this question by considering a variant of the problem where we allow each coin to be flipped an unlimited number of times.
We call this variant the \textit{Unlimited-Flips Unanimous Vote} problem. 
We now show that in this variant it is optimal to use the greedy rule at every position; there is no exceptional position. Since each coin may be flipped an unlimited number of times, we assume strict inequality between coin parameters (\(0\leq p_1 < \ldots < p_n\leq 1\)) in what follows. We use the notation \(\bar{p}_i\eqdef 1 - p_i\).
We represent an infinite sequence of the $n$ coins as a function \(\sigma : \mathbb{Z}_+\to [n]\). Omitted proofs in this section are in Appendix~\ref{app: missing section 2}.

\begin{definition}\label{def: cost}
    For sequence \(\sigma : \mathbb{Z}_+\to [n]\), \(\cost(\sigma)\) is the random variable that counts the number of flips until seeing both a head and a tail if the \(n\) coins are flipped in the order specified by \(\sigma\).
\end{definition}

\begin{definition}
    For sequence \(\sigma : \mathbb{Z}_+\to[n]\) and \(i \leq j\), define 
    \[
        H_\sigma(i, j) \eqdef \prod_{k=i}^{j}p_{\sigma(k)}, \quad T_\sigma(i, j) \eqdef \prod_{k=i}^{j}\bar{p}_{\sigma(k)}
    \]
\end{definition}

\(H_\sigma(i,j)\) and \(T_\sigma(i,j)\) are the probabilities that the coin flips at positions \(i\) through \(j\) are all heads or all tails, respectively. Abusing notation, we use \(H_\sigma(j)\) (resp. \(T_\sigma(j)\)) to denote \(H_\sigma(1,j)\) (resp. \(T_\sigma(1,j)\)); we will drop the subscript when the sequence is clear from context. We can use this notation to write the expected cost of a sequence as follows:

\begin{fact}\label{lem: sequence cost}
    Let \(\sigma : \mathbb{Z}_+\to [n]\). Then
    \[
        \expect*{\cost(\sigma)} = 1 + \sum_{k=1}^\infty
        \left(
            H_\sigma(k) + T_\sigma(k)
        \right)
    \]
\end{fact}

\noindent
It is also easy to see that under minimal assumptions, an optimal sequence has finite expected cost.

\begin{fact}\label{lem: optimal is finite}
    An optimal sequence's expected cost is infinite if and only if \(n = 1\) and \(p_1\in\{0,1\}\).
\end{fact}

\noindent
In what follows, we assume the condition of \Cref{lem: optimal is finite} does not hold, so that an optimal sequence has finite expected cost.

\begin{lemma}\label{lem: only use 1 and n}
    There is an optimal sequence \(\sigma^*\) such that
    \(
        \forall k\in\mathbb{Z}_+,\ \sigma^*(k) \in\{1,n\}
    \).
\end{lemma}

\noindent
The main idea of the proof is that, given any sequence $\sigma$ such that $\sigma(k) \not\in \{1,n\}$ for some $k$, it is possible to iteratively construct a new sequence using only coins \(1\) and \(n\) with expected cost at most that of \(\sigma\). This is a straightforward consequence of the formula of \Cref{lem: sequence cost}.\par

Since there is an optimal sequence using only the two coins with the lowest and highest parameters, we can reduce any instance to an instance with two coins: \(p_1\) and \(p_n\). We can then show, using an exchange-argument proof similar to that of \cite{dumankeles2026}, that there is a greedy sequence that is optimal.

\begin{theorem}\label{thm: greedy optimal}
    There is an optimal sequence for the Unlimited-Flips Unanimous Vote problem that only uses coin \(1\) if \(\frac{1}{2} < p_1 < p_n\) and only uses coin \(n\) if \(p_1 < p_n < \frac{1}{2}\). If \(p_1 \leq \frac{1}{2} \leq p_n\), there is an optimal sequence that uses coin \(1\) at 1-biased positions and coin \(n\) at 0-biased and unbiased positions.
\end{theorem}

\begin{proof}[Proof sketch:]
    Using \Cref{lem: only use 1 and n}, we can restrict to the case where there are two coins \(p_1 < p_n\). We consider two cases:
    \begin{enumerate}
        \item \(p_1 < p_n < \frac{1}{2}\)
        or \(\frac{1}{2}< p_1 < p_n\)
        \item
        \(p_1 \leq \frac{1}{2}\leq p_n\)
    \end{enumerate}
    For Case 1, take \(p_1 < p_n < \frac{1}{2}\), since \(\frac{1}{2} < p_1 < p_n\) can be handled symmetrically. It is straightforward to show (using the formula of \Cref{lem: sequence cost}) that if a sequence has coin \(1\) in some position, the expected cost can be strictly improved by replacing this choice with coin \(n\).\par
    For Case 2, we use a similar argument. If, say, coin \(1\) is used at a 0-biased position (call this position \(i\)), then if coin \(n\) occurs later in the sequence, say, next at position \(j\), we can swap positions \(i\) and \(j\) and improve the expected cost. If \(n\) never occurs again (since we are considering an infinite sequence), we can simply change the coin at the considered position to \(n\), which is similar to Case 1.
\end{proof}

\subsection{Adaptivity Gap}

\cite{dumankeles2026} was also concerned with the \textit{adaptivity gap} of the Unanimous Vote problem, which is the supremum, over all instances, of the ratio of the optimal non-adaptive cost to the optimal adaptive cost. Here, we are also concerned with the \textit{additive adaptivity gap}, which is defined similarly, except using the \textit{difference} between the optimal non-adaptive cost and the optimal adaptive cost. \cite{dumankeles2026} showed that the adaptivity gap of the Unanimous Vote problem is \(1.2 \pm \mathcal{O}\left(\frac{1}{2^n}\right)\), and their work implicitly bounds the additive adaptivity gap to be in \(\left[\frac{1}{2}, 1\right]\). For the Unlimited-Flips Unanimous Vote problem, we resolve both of these questions.

\begin{theorem}\label{thm: adgap}
    The adaptivity gap of the Unlimited-Flips Unanimous Vote problem is \(1.2\).
\end{theorem}

\begin{theorem}\label{thm: add adgap}
    The additive adaptivity gap of the Unlimited-Flips Unanimous Vote problem is \(1/2\).
\end{theorem}

\noindent
The lower bound in both cases comes from the instance \(p_1 = 0\), \(p_2 = \frac{1}{2}\), where the optimal adaptive cost is \(\frac{5}{2}\), while the optimal non-adaptive cost is \(3\). The main idea of the proof of the upper bound is that choosing the best out of a small number of candidate sequences always gives the desired bound against the optimal adaptive strategy. To reduce the number of strategies we consider, we first assume without loss of generality that \(p_1 \leq \bar{p}_n\). For the (multiplicative) adaptivity gap, we use the sequences
\begin{itemize}
    \item \(\sigma_1: n,n,n,n,\ldots\) (flips only coin \(n\))
    \item \(\sigma_2 : 1,n,n,n,\ldots\) (flips coin \(1\) then only \(n\))
    \item \(\sigma_3 : n,1,n,1,\ldots\) (repeatedly alternates between coin \(n\) and \(1\))
\end{itemize}
For the additive adaptivity gap, we replace \(\sigma_3\) with the strategy \(\sigma_4: 1,n,1,n,n,\ldots\) which alternates twice between coin \(1\) and \(n\), then continues with only coin \(n\). The cost of each of these strategies (as well as the optimal adaptive strategy) has a closed-form expression, and the proof reduces to proving inequalities between rational functions on a subset of \([0,1]^2\). The proofs of these inequalities we give as Lean~\cite{moura2021} files, linked in Appendix~\ref{app: adgap}. For a visualization of the relevant inequalities, see this \DesmosLink{Desmos graph}.

\subsection{Periodicity}\label{sec: periodicity}

In this section, we consider the periodicity of optimal sequences. If \(p_1 < p_n \leq 1/2\) or \(1/2 \leq p_1 < p_n\), there is an optimal sequence that only uses one coin. If \(p_1 = 0 < 1/2 < p_n\), the sequence \(1,n,n,n,\ldots\) is optimal (and symmetrically if \(p_n=1>1/2>p_1\) then \(n,1,1,1,\ldots\) is optimal). So, throughout the rest of this section we assume \(0<p_1 < 1/2 < p_n<1\). We consider specifically the optimal sequence that uses \(n\) at unbiased positions, and refer to it as ``the'' optimal sequence. Denote this sequence by \(\sigma^* : \mathbb{Z}_+\to\{1,n\}\). In what follows, we use the term ``periodic'' to refer to pure periodicity; that is, we say that a sequence \(\sigma\) is periodic iff there is a finite word \(s\) such that \(\sigma = sss\ldots\). We use the term ``aperiodic'' to mean that a sequence is not even eventually periodic.\par
We begin by showing that the optimal sequence corresponds closely to a particular \textit{mechanical word}.
\begin{definition}[Lower Mechanical Word, see 2.1.2 in \cite{lothaire2002}]\label{def: mechanical word}
    Given two real numbers \(\alpha\) and \(\rho\) with \(0 \leq \alpha \leq 1\), we define the infinite word \(
        s_{\alpha,\rho} : \mathbb{Z}_{\geq 0}\rightarrow \{0,1\}
    \) by
    \begin{align*}
        s_{\alpha,\rho}(m) &= \lfloor\alpha(m + 1) + \rho\rfloor - \lfloor\alpha m + \rho\rfloor
    \end{align*}
    Since \(0 \leq \alpha \leq 1\), \(s_{\alpha,\rho}(m)\in \{0,1\}\). \(s_{\alpha,\rho}\) is called the lower mechanical word with slope \(\alpha\) and intercept \(\rho\).
\end{definition}

\begin{definition}
    Define $h \eqdef \frac{p_n}{1-p_n}$ \text{ and } $ t \eqdef \frac{1-p_1}{p_1}$.
\end{definition}

\noindent
In what follows, suppose \(h > t\), or equivalently, $p_1 > 1-p_n$.  If \(h < t\), the proof is symmetric (redefining \(\sigma^*\) to use \(1\) at unbiased positions). If \(h = t\), then \(p_1 = \bar{p}_n\) and it is easy to check that the alternating sequence \(n,1,n,1,\ldots\) is optimal.
\begin{lemma}\label{lem: nmax}
    For any \(j\in\mathbb{Z}_{+}\), let \(\nmax(j)\) and \(\nmin(j)\) denote the number of occurrences of \(n\) and \(1\) at or prior to position \(j\) of \(\sigma^*\), respectively. The following three statements are equivalent:\\ (i) \(\sigma^*(j) = n\), (ii) \(H_{\sigma^*}(j-1)\leq T_{\sigma^*}(j-1)\), and (iii) \(\nmax(j-1)\leq \nmin(j-1)\log_ht\)
\end{lemma}

\noindent
\textit{Proof.}
    Equivalence of (i) and (ii) follows easily from the behavior of the optimal sequence. We show equivalence of (ii) and (iii). Note that \(h,t> 1\). Using \(H(j)\) (resp. \(T(j)\)) in place of \(H_{\sigma^*}(j)\) (resp. \(T_{\sigma^*}(j)\)), 
    \begin{align*}
        H(j) &\leq T(j) 
        \Leftrightarrow p_{n}^{\nmax(j)}p_{1}^{\nmin(j)} \leq (1-p_n)^{\nmax(j)}(1-p_1)^{\nmin(j)} 
        \Leftrightarrow 
        \left(\frac{p_n}{1-p_n}\right)^{\nmax(j)} \leq \left(\frac{1-p_1}{p_1}\right)^{\nmin(j)}\\
        \Leftrightarrow h^{\nmax(j)} &\leq t^{\nmin(j)} \Leftrightarrow \nmax(j) \leq  \nmin(j)\log_ht\tag*{$\square$}
    \end{align*}

\noindent
Informally, the following lemma states that \(n\) never occurs twice in a row in \(\sigma^*\).
\begin{lemma}\label{lem: no double xmax}
    For any \(k\in\mathbb{Z}_+\), we have \(
        \sigma^*(k) = n \Rightarrow \sigma^*(k + 1) = 1
    \).
\end{lemma}

\begin{proof}
    Suppose, for the sake of contradiction, that \(\sigma^*(k) = \sigma^*(k+1) = n\). Take \(k\) to be minimal. If \(k = 1\), then we have \(\nmin(1)\log_ht = 0 < 1 = \nmax(1)\). This implies that \(\sigma^*(2) = 1\) by Lemma \ref{lem: nmax}, and we are done. So, assume \(k\neq 1\), meaning that there is at least one occurrence of \(1\) before position \(k\) (since \(k\) is minimal). Equivalently, \(\nmin(k) \geq 1\). By Lemma \ref{lem: nmax}, at indices \(k\) and \(k+1\) we have that $\nmax(k-1) \leq \nmin(k-1) \log_ht$ and $\nmax(k) \leq \nmin(k) \log_ht$.
    Also \(\nmax(k) = \nmax(k-1) + 1\) and \(\nmin(k-1) = \nmin(k)\). Thus
    \[
        \nmax(k-1) + 1 \leq \nmin(k - 1)\log_ht \Rightarrow \nmax(k-1) \leq (\nmin(k-1) - 1)\log_ht
    \]
    where the implication follows from \(0 < \log_ht < 1\). We know that \(\nmin(k-1) = \nmin(k)\geq 1\). Since \(k\) is minimal, we know that \(\sigma^*(k-1) = 1\), and therefore that \(\nmin(k-1) - 1= \nmin(k-2)\). Since \(\sigma^*(k-1) = 1\), we also have that \(\nmax(k-1) = \nmax(k-2)\). Substituting, we obtain $\nmax(k-2) \leq \nmin(k-2)\log_ht$. This contradicts the fact that coin \(1\) was chosen for position \(k - 1\).
\end{proof}

\begin{lemma}\label{lem: index count}
Define \(j_k\) to be the index of the \(k\)\textsuperscript{th} occurrence of \(n\) in \(\sigma^*\). Then
\[
    \nmin(j_k) = \min\{m\in \mathbb{Z}_{\geq 0}: m\log_ht\geq k  - 1\}
\]
\end{lemma}

\begin{proof}
    Note that for \(j_k=1\), we have \(k = 1\), and the statement follows by direct computation. For the case \(j_k = 2\), we know \(\sigma^*(1)=n\) by definition, so \(k = 2\) which is impossible by \Cref{lem: no double xmax} (since we would have \(\sigma^*(1)=\sigma^*(2)=n\)). So, assume \(j_k > 2\).\par
    Recall that \(\nmin(j_k)\) is the number of occurrences of \(1\) prior to the \(k\)\textsuperscript{th} occurrence of \(n\). We know that \(\nmax(j_k-1) \leq \nmin(j_k-1)\log_ht\), by Lemma \ref{lem: nmax}. We also know that \(\nmax(j_k-1) = k - 1\) by the definition of \(j_k\). Substituting, we then have that \(\nmin(j_k-1)\log_ht \geq k-1\). We also know that, since \(\sigma^*(j_k) = n\), \(\nmin(j_k-1) = \nmin(j_k)\). Substituting, \(\nmin(j_k)\log_ht\geq k - 1\).\par
    Suppose for the sake of contradiction that \(\nmin(j_k)\) is not minimal while satisfying \(\nmin(j_k)\log_ht \geq k - 1\). We know that \(\sigma^*(j_k - 1) = 1\) by \Cref{lem: no double xmax}, so
    \[
        \nmin(j_k) = \nmin(j_k - 1) = \nmin(j_k-2) + 1\Leftrightarrow \nmin(j_k) - 1= \nmin(j_k-2)
    \]
    By our assumption then, \(\nmin(j_k-2)\log_ht \geq k - 1 = N_{n}(j_k-1) = N_n(j_k - 2)\). By Lemma \ref{lem: nmax}, \(\sigma^*(j_k -1) = n\), which contradicts \(\sigma^*(j_k - 1) = 1\).
\end{proof}

\begin{lemma}\label{lem: ht lemma}
    The optimal sequence is periodic if and only if the word \(s_{\log_ht,0}\) is periodic.
\end{lemma}

\begin{proof}
    Recall that \(h > t\). For ease of notation, write \(s\) in place of \(s_{\log_ht,0}\). Let \(i_k\) be the index of the \(k\)\textsuperscript{th} occurrence of \(1\) in \(s\). We have, for \(k > 0\)
    \begin{align*}
        i_k &= \min\{m\in \mathbb{Z}_{\geq0} : \lfloor(m+1)\log_ht\rfloor = k\}\\
        &= \min\{m\in\mathbb{Z}_{\geq0}: (m+1)\log_ht \geq k\}\hspace{10pt}\text{since \(0 < \log_ht < 1\)}
    \end{align*}
    Define \(i_0 = -1\). Define \(j_k\) to be the index of the \(k\)\textsuperscript{th} occurrence of \(n\) in \(\sigma^*\). By Lemma \ref{lem: index count}, we have $\nmin(j_k) = \min\{m\in \mathbb{Z}_{\geq 0}\mid m\log_ht\geq k  - 1\}$
    so \(\nmin(j_{k+1}) = i_k + 1\). Consider the number of 0s between the \((k - 1)\)\textsuperscript{th} and \(k\)\textsuperscript{th} 1 in \(s\), denote this value by \(z_k\). We define \(z_1\) to be the number of 0s prior to the first 1 in \(s\). We have $z_k = i_k - i_{k-1} - 1$.
    Similarly, let \(y_k\) be the number of \(1\)s in \(\sigma^*\) between the \(k\)\textsuperscript{th} and \((k+1)\)\textsuperscript{th} occurrence of \(n\). We have
    \[
        y_k = \nmin(j_{k+1}) - \nmin(j_{k}) = (i_k + 1) - (i_{k-1} + 1) = i_k - i_{k-1} = z_k + 1
    \]
    Equivalently, the number of \(1\)s between the \(k\)\textsuperscript{th} and \((k + 1)\)\textsuperscript{th} occurrence of \(n\) in \(\sigma^*\) is one more than the number of \(0\)s between the \((k-1)\)\textsuperscript{th} and \(k\)\textsuperscript{th} 1 in \(s\). Recall that by \Cref{lem: no double xmax}, \(n\) never occurs twice in a row in \(\sigma^*\), so that its periodicity is entirely determined by the length of contiguous choices of coin \(1\). Thus \(\sigma^*\) is periodic if and only if \(s\) is periodic.
\end{proof}

\begin{lemma}[weak version of Lemma 2.1.14 in \cite{lothaire2002}]\label{lem: word equivalences}
    Let \(s\) be a mechanical word with slope \(\alpha\). If \(\alpha\) is rational then \(s\) is purely periodic. If \(\alpha\) is
    irrational, then \(s\) is aperiodic.
\end{lemma}

\begin{theorem}
    The optimal sequence is periodic if and only if \(\log_ht\) is rational.
\end{theorem}

\begin{proof}
    \(\log_ht\) is rational if and only if the word \(s_{\log_ht,0}\) is periodic, by \Cref{lem: word equivalences}. The word \(s_{\log_ht,0}\) is periodic if and only if the optimal sequence is periodic, by \Cref{lem: ht lemma}.
\end{proof}

\section{\textit{d}-ary Unanimous Vote}\label{sec: duv}

We now consider a generalization that replaces the Unanimous Vote problem's coins with \(d\)-sided dice, which we call the \textit{\(d\)-ary Unanimous Vote} problem. Specifically, given \(n\) independent, biased, \(d\)-sided dice numbered 1 through $n$, 
where the sides of each die are colored using \(d\) colors, the \(d\)-ary Unanimous Vote problem asks us to fix a permutation \(\pi:[n]\rightarrow[n]\) by which to roll the dice, stopping when two distinct colors are rolled, or all dice have been rolled. The goal is to minimize the expected number of rolls.

The rest of this section uses the following notation.
Let $X_i$ be the random variable whose value is the outcome of rolling die $i$.
Define \(\cost(\pi)\) as a discrete random variable that is equal to the number of rolls that are  performed if they are done in the order specified by \(\pi\) (roll die \({\pi(1)}\) first, then die \({\pi(2)}\), ...). Let the set of colors be \(\Omega\) (\(|\Omega|=d\)). For \(c\in\Omega\), write \(c_i\eqdef\Pr[X_i=c]\), \(\bar{c}_i\eqdef 1- c_i\).  For fixed permutation \(\pi\), define a corresponding function \(P_\pi(k;c) \eqdef \prod_{i=1}^kc_{\pi(i)}\), which is the probability that the first \(k\) dice of \(\pi\) show color \(c\). If \(\pi\) is clear from context, we may drop the subscript. Omitted proofs from this section are in Appendix~\ref{app: missing duv section}.

\subsection{1-Additive Approximation Algorithm}

We first state our main result:

\begin{theorem}\label{thm: greedy less than opt plus one}
    Let \(\OPT\) be the expected cost of an optimal adaptive strategy for the \(d\)-ary Unanimous Vote problem. There is a polynomial-time algorithm that produces a permutation \(\pi\) such that
    \[
        \expect*{\cost(\pi)} \leq \text{\textnormal{OPT}} + 1
    \]
\end{theorem}

Our algorithm for the \(d\)-ary Unanimous Vote problem, in the case \(d = 2\), is almost exactly the greedy Algorithm 1 of \cite{dumankeles2026}; the only difference is step \(1\) of our algorithm, which does not have an analogue in theirs. The greedy rule of \cite{dumankeles2026} can be viewed in two ways, which are equivalent for \(d = 2\):
\begin{itemize}
    \item Greedy rule A: choose a die that has maximal probability of having a different result from the previously chosen dice, conditioned on all previously chosen dice being equal.
    \item Greedy rule B: choose a die that has minimal probability of a color that has maximal bias at the considered position.
\end{itemize}
To clarify greedy rule B, we can generalize the notion of bias at position \(k\) to mean the color that has a maximum likelihood of the first \(k-1\) rolls identically showing this color (biases may be tied among a subset of the colors). We use \textit{both} of these greedy rules to construct the prefix and suffix of the permutation, respectively, separated by a position called the \textit{crossover point}.

\begin{definition}[\cite{dumankeles2026}]
    Let \(\pi\) be a permutation. Define \(t\geq 2\), the \textnormal{crossover point} of \(\pi\), to be minimal such that, for any \(\pi'\) identical to \(\pi\) in positions \(1,\ldots,t-1\), we have
    \[
        \Pr[\cost(\pi') > t \mid \cost(\pi') > t - 1] > \frac{1}{2}
    \]
\end{definition}

\noindent
Note first that such a \(t\) may not always exist, depending on \(\pi\) and the distribution of the dice. We will sometimes refer to the crossover point \(t\) of \(\pi\) when only a prefix of \(\pi\) has been defined; in that case, note that if there is a position satisfying the inequality of the definition, \(t\) is independent of the suffix of \(\pi\). The idea of the crossover point was originally used to show the 1.5-approximation of Claim 6.3 in \cite{dumankeles2026}, which they use to upper bound the multiplicative adaptivity gap. The analysis in \cite{dumankeles2026} can also be used to prove a 1-additive approximation bound, with respect to the optimal adaptive strategy, for the \(d = 2\) case. Although our algorithm, in the case \(d=2\), is almost exactly Algorithm 1 in \cite{dumankeles2026}, our analysis bears more similarity to the proof of the 1.5-approximation of Claim 6.3.\par
We will need the following lemma before stating the algorithm. This lemma will also be key to the analysis.
\begin{lemma}\label{lem: properties after t}
    Let \(\pi\) be a permutation with a position \(z\geq 2\) such that, for any \(\pi'\) identical to \(\pi\) in positions \(1,\ldots,z-1\), \(\Pr[\cost(\pi')>z\mid\cost(\pi')>z-1] > 1/2\). There exists \(r\in\Omega\) such that
    \begin{enumerate}
        \item \(\frac{P_\pi(z-1;r)}{\sum_{c\in\Omega}P_{\pi}(z-1; c)}\geq \frac{1}{2}\)
        \item For any remaining die \(i\in[n]\backslash\pi([z-1])\), \(r_i\geq\frac{1}{2}\)
    \end{enumerate}
\end{lemma}

\noindent
Note that taking \(z\) to be minimal gives the properties for the crossover point \(t\) in particular. Property (1) can be derived by algebraic manipulation from the inequality \(\Pr[\cost(\pi)>z\mid\cost(\pi)>z-1] > 1/2\), and property (2) can be proven using property (1) along with the inequality \(\Pr[\cost(\pi)>z\mid\cost(\pi)>z-1] > 1/2\). We can now state our algorithm; see \Cref{alg: greedy}.

\begin{algorithm}
\caption{1-Additive Approximation Algorithm for \(d\)-ary Unanimous Vote}\label{alg: greedy}
\vspace{2pt}
\begin{enumerate}
    \item If there is \(b\in\Omega\) such that \(b_i \geq \frac{1}{2}\) for all \(i\in[n]\), sort the dice in non-decreasing order of \(b_i\) and return.
    \item Construct a permutation \(\pi\), starting from position 1, by repeatedly choosing the next die using greedy rule A, until either all dice have been chosen or placing the chosen die in the next position would cause it to become the crossover point \(t\).
    \item If \(t\) was reached, let \(r\) be the color guaranteed by \Cref{lem: properties after t}; from position \(t\) forward, add all dice not already chosen (for positions \(1,\ldots,t-1\)) in non-decreasing order of \(r_i\).
    \item Return \(\pi\).
\end{enumerate}
\vspace{-8pt}
\end{algorithm}

\noindent
Note that because of the properties in \Cref{lem: properties after t}, step 3 of \Cref{alg: greedy} is precisely greedy rule B. Note also that in step 2, if placing the die chosen by greedy rule A into the next position would cause it to become the crossover point \(t\), then by the definition of greedy rule A, placing any other die in that position would also cause it to become the crossover point. Thus \(t\) is, in fact, the crossover point of the final constructed permutation \(\pi\). Finally, note that on step 2 the algorithm can use any first die; this is intentional.\par
The main idea behind the proof of the 1-additive bound is that before the crossover point, the expected cost corresponds to a partial geometric series of powers of \(\frac{1}{2}\). After the crossover point, we may incur a large amount of cost in expectation, but much of this cost is also incurred by an optimal adaptive strategy for a relaxed problem: finding some \(X_i\neq r\) while rolling at least two dice (where \(r\) is the color guaranteed by \Cref{lem: properties after t}). This strategy has expected cost \(\OPT_R\), a lower bound on \(\OPT\). The cost not charged to \(\OPT_R\) forms the tail of the geometric series, which altogether sums to at most \(1\). \Cref{thm: greedy less than opt plus one} directly implies the following corollary.
\begin{corollary}\label{thm: adgap upper bound}
    The adaptivity gap of the \(d\)-ary Unanimous Vote problem is at most \(\frac{3}{2}\), and its additive adaptivity gap is at most \(1\).
\end{corollary}
\noindent
The multiplicative result follows from \(\OPT\geq 2\), so \(\OPT + 1 = \OPT\left(1 + \frac{1}{\OPT}\right)\leq \frac{3}{2}\OPT\).

\subsection{Polynomial-Time Approximation Scheme and Hardness}

We can extend the ideas of \Cref{thm: greedy less than opt plus one} to design a PTAS. The main idea is to brute force the first \(\lceil1/\varepsilon\rceil\) positions of the permutation, then run \Cref{alg: greedy} on the remaining portion.

\begin{theorem}\label{thm: d-ary ptas}
    There is a polynomial-time approximation scheme for the \(d\)-ary Unanimous Vote problem.
\end{theorem}

\noindent
See Appendix \ref{app: d-ary ptas} for the proof of \Cref{thm: d-ary ptas}. The most natural next question is whether the problem can be solved exactly in polynomial time; in Appendix \ref{app: d-ary np-hard} we show that if \(d\) is part of the input, solving the problem is \(\mathsf{NP}\)-hard. The main idea is to correspond \(\Pr[X_i=c]=0\) with die \(i\) ``covering'' color \(c\), and reduce from Exact Cover by 3-Sets.

\subsection{The Unlimited-Rolls Variant}

It is natural to consider the analogue of the Unlimited-Flips extension in the \(d\)-ary case: given a collection of \(n\) dice, each of which can be rolled \textit{arbitrarily many} times, what sequence of dice minimizes the expected number of rolls until seeing two distinct colors? It is relatively easy to extend the idea of \Cref{thm: greedy less than opt plus one} to show that a greedy sequence (using one of greedy rules A and B) achieves a 1-additive approximation to the adaptive optimum of the Unlimited-Rolls variant.\par
That is, either there is some color \(r\) (as in \Cref{lem: properties after t}) such that all dice have \(r_i\geq \frac{1}{2}\), or there is not. If there is, the sequence \(\sigma_B\) using only a die \(i'\) with minimal \(r_i\) achieves a 1-additive approximation, since we can upper bound its cost as follows:
\begin{align*}
    \expect{\cost(\sigma_B)} = 2 + \sum_{k=2}^\infty r_{i'}^k + \sum_{c\in\Omega\backslash\{r\}}\sum_{k=2}^\infty c_{i'}^k
    \leq \underbrace{2 + \sum_{k=2}^\infty r_{i'}^k}_{\leq\OPT} + \underbrace{\sum_{k=2}^\infty \left(1-r_{i'}\right)^k}_{\leq\frac{1}{2}}\leq \OPT + \frac{1}{2}
\end{align*}
If there is no such \(r\in\Omega\), we have at least a \(\frac{1}{2}\) probability of stopping on each roll, conditioned on that roll occurring (by an almost identical version of \Cref{lem: properties after t}). The expected cost in this case is at most \(2 + \frac{1}{2} + \frac{1}{4} + \ldots = 3\), and noting that \(\OPT\geq 2\) gives the result. This leads to the following corollary:
\begin{corollary}\label{cor: unlimited-roll adaptivity gap upper bound}
    The adaptivity gap of the Unlimited-Rolls \(d\)-ary Unanimous Vote problem is at most \(\frac{3}{2}\), and its additive adaptivity gap is at most \(1\).
\end{corollary}

\noindent
We conjecture that the exact bounds for the adaptivity gap and additive adaptivity gap for arbitrary \(d\) match the \(d = 2\) case: \(1.2\) and \(\frac{1}{2}\), respectively.

It is also natural to wonder whether either of greedy rules A or B is optimal for the Unlimited-Rolls variant, since they are both optimal (and equivalent) for the Unlimited-Flips variant. It is easy to construct counterexamples to both rules; we do so in Appendix~\ref{app: counterexamples}.

\section{Finite Coupon Collection Problem}\label{sec: fccp}

The Unanimous Vote problem can also be generalized in the following way: instead of 
rolling dice in a fixed order until we either get two distinct colors or have rolled all dice, we roll the dice in a fixed order until we have either gotten (collected) {\em all} $d$ distinct colors, or have rolled all dice. We call this the
\textit{Finite Coupon Collection} problem. We give an \(\mathcal{O}(\log d)\)-approximation algorithm for this problem.

The idea behind the algorithm is to reduce to the Submodular Ranking problem of Azar and Gamzu~\cite{azar2011}. Submodular Ranking takes as input a ground set \([n]\) and a collection of $m$ monotone, submodular set functions \(f^i:2^{[n]}\to [0,1]\) for \(i=1,\ldots,m\). Each function must satisfy \(f^i(\varnothing) = 0\) and \(f^i([n]) = 1\). 
Associated with each function \(f^i\) is a weight \(w_i\in \mathbb{R}_+\). 
The goal is to output a permutation of the elements \(\pi : [n]\to[n]\) that minimizes the weighted sum of the ``cover times'' of the functions. That is, $\pi$ must minimize the following sum: 
\begin{equation}
\label{eq:submodrankcost}
    \sum_{i=1}^mw_i\min\{1\leq j\leq n : f^i(\{\pi(1),\ldots,\pi(j)\}) = 1\}
\end{equation}
Azar and Gamzu presented an \(\mathcal{O}(\log1/\varepsilon)\)-approximation algorithm for Submodular Ranking, where
\[
    \varepsilon\eqdef \min_{
        S\subseteq[n],\ 1\leq j\leq n,\ 1\leq i\leq m
    }\{f^i(S\cup\{j\}) - f^i(S) : f^i(S\cup\{j\}) > f^i(S) \}
\]
In words, \(\varepsilon\) is the smallest positive marginal value a function $f^i$ can gain from any single element. Their algorithm is greedy: 
for the $k$\textsuperscript{th} item in the permutation, where $S$ is the set of $k-1$ items already chosen, they choose the item $j$ maximizing the following quantity:
\begin{align}
\label{agscore}
    \sum_{i:f^i(S)<1}w_i~\frac{f^i(S\cup\{j\})-f^i(S)}{1-f^i(S)}
\end{align}

\begin{theorem}\label{thm: fccp approximation}
    There is a polynomial-time \(\mathcal{O}(\log d)\)-approximation algorithm for the Finite Coupon Collection problem.
\end{theorem}
\noindent
We define the following notation to use in the proof of \Cref{thm: fccp approximation}.
\begin{definition}
   Let $X_1,\ldots,X_n$ be independent random variables corresponding to the outcomes of the rolls of dice $1, \ldots, n$ respectively.
    A {\em realization} of the $X_i$ is a vector $\alpha \in [d]^n$ which corresponds to having $X_i=\alpha_i$ for each $i \in [n]$. For $\alpha \in [d]^n$ and  \(S\subseteq [n]\) define
    \[
        \kappa(\alpha,S)\eqdef \bigcup_{i\in S}\{\alpha_i\}
    \]
\end{definition}
\begin{definition}
    For a permutation \(\pi\) on \([n]\) and realization \(\alpha \in [d]^n\), define
    \[
        \cost(\pi,\alpha)\eqdef \min
        \left\{
            1\leq j \leq n : 
                \bigcup_{i=1}^j\left\{\alpha_{\pi(i)}\right\}
            = [d] \textnormal{ or }j = n
        \right\}
    \]
\end{definition}

\noindent
That is, \(\cost(\pi,\alpha)\) is the number of rolls that will be performed on realization $\alpha$ (when each $X_i=\alpha_i$), if 
dice are rolled in the order specified by \(\pi\) until either all colors are seen or all dice have been rolled.

\noindent
\textit{Proof of \Cref{thm: fccp approximation}:}
To introduce our basic approach, we first describe an exponential-time algorithm that achieves an $\mathcal{O}(\log n)$ approximation bound. We then describe how to modify it to get an approximation bound of $\mathcal{O}(\log d)$ that runs in polynomial time.

We reduce the Finite Coupon Collection problem to Submodular Ranking as follows. We construct, for each realization $\alpha \in [d]^n$, a monotone submodular function $f^{\alpha}:2^{[n]} \rightarrow [0,1]$. Let \(\mathcal{G}\) be the subset of realizations containing all colors, i.e., $\mathcal{G} = \{\alpha \in [d]^n : \kappa(\alpha,[n])=[d]\}$. For $\alpha \in \mathcal{G}$, we define $f^{\alpha}(S) \eqdef \frac{|\kappa(\alpha,S)|}{d}$.  For $\alpha \not\in \mathcal{G}$, we define $f^{\alpha}(S) \eqdef \frac{|S|}{n}$. For each realization $\alpha \in [d]^n$, let $w_{\alpha}$ be the probability of realization $\alpha$ in the Finite Coupon Collection instance. The ground set $[n]$, $d^n$ functions $f^{\alpha}$, and associated weights $w_{\alpha}$ constitute the instance of Submodular Ranking.

By the definition of $f^{\alpha}$, for $\alpha \in \mathcal{G}$, $f^{\alpha}(S) = 1$ iff for realization $\alpha$, subset $S$ of dice includes all colors ($\kappa(\alpha,S)=[d]$).  For $\alpha \not\in \mathcal{G}$, $f^{\alpha}(S) = 1$ iff $S = [n]$. 
In either case, $f^{\alpha}(S)=1$ iff in the Finite Coupon Collection instance, rolling all dice in $S$ on realization $\alpha$ would achieve the stopping condition.  

Clearly, for any permutation $\pi$ of $[n]$, the value of the Submodular Ranking objective function in (\ref{eq:submodrankcost}) for this instance is equal to the cost of $\pi$ for the original Finite Coupon Collection instance.  Thus running the $\mathcal{O}(\log 1/\varepsilon)$-approximation algorithm of Azar and Gamzu on our resulting Submodular Ranking instance yields an $\mathcal{O}(\log 1/\varepsilon)$-approximation to the optimal Finite Coupon Collection cost. Because of the way we defined the $f^{\alpha}$ for $\alpha \not\in \mathcal{G}$, $1/\varepsilon = n$. Therefore, we now have an exponential-time $\mathcal{O}(\log n)$-approximation algorithm for Finite Coupon Collection.

To improve the approximation factor to $\mathcal{O}(\log d)$, we note that for all $\alpha\notin\mathcal{G}$ and any permutation $\pi$, $\cost(\pi,\alpha)=n$. That is, the expected costs of permutations $\pi$ differ only according to their cost for realizations $\alpha \in \mathcal{G}$. Thus in the above reduction to Submodular Ranking, it is unnecessary to include the $f^{\alpha}$ for $\alpha \not\in \mathcal{G}$. One can omit the $f^{\alpha}$ with $\alpha \not\in \mathcal{G}$ when forming the Submodular Ranking instance and normalize the weights so that \(w_\alpha=\Pr[\alpha\mid \alpha\in\mathcal{G}]\); running the algorithm of \cite{azar2011} on this smaller instance will still yield an $\mathcal{O}(\log 1/\varepsilon)$-approximation for the Finite Coupon Collection instance. The difference is that now $\varepsilon = 1/d$, because its value depends only on the $f^{\alpha}$ with $\alpha \in \mathcal{G}$.


To achieve polynomial runtime we use sampling. Instead of explicitly computing the score (\ref{agscore}) for each element, for each greedy step we use a polynomially-sized sample of randomly generated realizations to estimate the score of every remaining element. We use rejection sampling to ensure \(\alpha\in\mathcal{G}\) for the realizations used in the estimate. We then choose the element with the highest estimated score.

Because we are only estimating the score (\ref{agscore}) of each element in each step, our choice is only approximately greedy. To complete the proof, it is therefore necessary to show that this approach will still yield an \(\mathcal{O}(\log 1/\varepsilon)\)-approximation. We omit the remaining details here and present them in Appendix~\ref{app: sampling}; the arguments are very similar to the arguments used in Appendix A.2 of~\cite{ghuge2022}, which is concerned with an almost identical difficulty. Our proof relies on the fact that the algorithm of \cite{azar2011} still achieves an $\mathcal{O}(\log 1/\varepsilon)$-approximation if the estimates of the scores are within a constant factor of their true values.
\hfill\(\square\)

\section{Generalizing the Reduction to Submodular Ranking}\label{sec: reduction generalization}
Our reduction to Submodular Ranking in the previous section can be described more generally as a generic approach to designing approximation algorithms for non-adaptive stochastic probing problems, such as non-adaptive SBFE.  In particular, it can be used for problems where there are $n$ probes, each probe has a random outcome in $[d]$ and can be performed at most once, and the goal is to minimize the expected number of probes until some stopping condition is reached. 

\par
The approach requires that you define, for any given realization $\alpha \in [d]^n$, a monotone submodular function $f^{\alpha}:2^{[n]} \rightarrow [0,1]$ such that $f^\alpha(\varnothing)=0$ and for all $S \subseteq [n]$, $f^{\alpha}(S) = 1$ iff on realization $\alpha$, probing the $i$ in $S$ achieves the stopping condition. It must be possible to evaluate $f^{\alpha}$ on any input $S$ in polynomial time, so that you can compute the (approximately) greedy choice at each step using the \(f^\alpha\) corresponding to a polynomially-sized set of sampled realizations \(\alpha\). The resulting approximation bound for the non-adaptive probing problem is $\mathcal{O}(\log 1/{\varepsilon})$, so the quality of the bound depends crucially on the magnitude of $\varepsilon$.
This approximation bound is with respect to the optimal {\em non-adaptive} strategy.

This approach is analogous to the $Q$-value approach introduced by Deshpande et al.\ as a method for obtaining approximation algorithms for adaptive SBFE problems \cite{deshpande2016}. The \(Q\)-value approach works by reducing SBFE to Stochastic Submodular Cover, through the construction of a single, (realization-dependent) monotone submodular utility function $g$ associated with the Boolean function being evaluated.
The function $g$ maps pairs $(S,\beta)$ to values in $[0,1]$, where $S \subseteq [n]$ is a set of items to be probed, and $\beta \in \Omega^S$ is a possible realization of those items. In SBFE problems, $\Omega \eqdef \{0,1\}$.
For a (full) realization $\alpha \in \Omega^{[n]}$ and $S \subseteq [n]$, let $\alpha_S$ be the partial realization setting each element $i \in S$ to $\alpha_i$.  Now, for $\alpha \in \Omega^{[n]}$, define $g^{\alpha}:2^{[n]} \rightarrow [0,1]$ to be such that $g^{\alpha}(S)=g(S,\alpha_S)$ for all $S \subseteq [n]$.
The function $g$ is required to have the following three properties: (1) $g(\varnothing,\zeta)=0$ where $\zeta$ is the empty realization, (2) for
all $\alpha \in \Omega^{[n]}$, $g^{\alpha}$ is a monotone submodular function, and (3) 
for all $S \subseteq [n]$ and $\beta \in \Omega^S$,
$g(S,\beta)=1$ iff after performing the probes in $S$, with outcomes according to $\beta$, no further probes would need to be performed (in the original SBFE instance).  

Once $g$ is constructed, the Adaptive Greedy algorithm of Golovin and Krause is used to solve the Stochastic Submodular Cover instance for $g$~\cite{golovin2011,hellerstein2021}, yielding an $\mathcal{O}(\log 1/\varepsilon)$ approximation bound for the original adaptive SBFE problem.  Here $\varepsilon$ is the smallest positive marginal value $g$ can gain from any single element.  The value $1/\varepsilon$ is essentially the same as the {\em goal value} quantity $Q$ defined in~\cite{deshpande2016}\footnote{Deshpande et al.\ actually define $g$ as mapping to a non-negative integer value, where for each realization $\alpha$ of all the elements in $[n]$, $g([n],\alpha)=Q$ for some fixed integer value $Q$. Their proof extends to the result as described here.}.
This $\mathcal{O}(\log 1/\varepsilon)$ approximation bound is with respect to the optimal {\em adaptive} strategy.

We observe that any $g$ used in the $Q$-value approach can also be used in our reduction to Submodular Ranking, by simply defining, for each realization $\alpha$ and subset \(S\), $f^{\alpha}(S) \eqdef g^{\alpha}(S)$.
Therefore, a consequence of our reduction technique is that $\mathcal{O}(\log 1/\varepsilon)$-approximation  bounds achieved for {\em adaptive} SBFE problems using the $Q$-value approach immediately imply $\mathcal{O}(\log 1/\varepsilon)$-approximation bounds for the analogous {\em non-adaptive} SBFE problems. For example, Deshpande et al.\ use the \(Q\)-value approach to achieve an $\mathcal{O}(\log T)$-approximation algorithm for the adaptive SBFE problem for decision tree representations of Boolean functions, where $T$ is the size of the decision tree. From the above, it follows that there is also an $\mathcal{O}(\log T)$-approximation algorithm for the non-adaptive version of the problem.

In contrast to the $Q$-value approach, 
which only gives an $\mathcal{O}(\log 1/\varepsilon)$ approximation bound for independent query distributions, our sampling-based reduction to Submodular Ranking for non-adaptive problems works for any distribution to which we have sample access.   
However, as with the $Q$-value approach, our approach only yields good approximation bounds if it is possible to construct appropriate submodular functions with small goal values.  This is not possible for all problems; see, e.g., cases considered in \cite{deshpande2016}.

\section{Open Questions}  
First, consider the $d$-ary Unanimous Vote problem for constant $d$. It is still open whether there is an efficient exact algorithm for this problem (the \(\mathsf{NP}\)-hardness result in Appendix~\ref{app: d-ary np-hard} does not apply to constant $d$).
It is also open whether there is a simple way to describe the optimal sequence for the Unlimited-Rolls version of this problem;
judging from the binary case, allowing for unlimited rolls should make the problem easier.

We also leave open the computational hardness of the Finite Coupon Collection problem.  Our Submodular Ranking-based approach cannot take advantage of the independence of the input distribution, and we conjecture that stronger results can be obtained with other techniques.

A generalization of Unanimous Vote not addressed in this paper is a variant of Finite Coupon Collection, with the change that we also stop rolling dice if it becomes impossible to collect all colors (for example, if there are two dice left to roll but three colors left to collect). This new stopping condition appears to make the problem much more challenging. We note that running our algorithm for the Finite Coupon Collection problem on this variant gives a \(\max\left\{\mathcal{O}(\log d), \frac{n}{n-d+2}\right\}\)-approximation, since the approximation ratio for realizations \(\alpha\in\mathcal{G}\) is \(\mathcal{O}(\log d)\), and any permutation trivially gives a \(\frac{n}{n-d+2}\)-approximation for \(\alpha\notin\mathcal{G}\).

In the \(d\)-ary Unanimous Vote problem, you stop rolling dice when you have seen two different colors.  What about a {\em Duplicate Detection} problem, where you instead stop rolling when you have seen the same color more than once? An easy observation is that for Duplicate Detection with $d=n$, the following tangentially related problem is at least as hard as computing the permanent of a matrix: compute the probability of getting all $n$ colors if you roll each of the $n$ dice.

Finally, our work on the Unlimited-Flips problem suggests a whole set of new open questions involving problems where each of $n$ probes can be performed an unlimited number of times, and the goal is to find an infinite probing sequence minimizing the expected number of probes until a certain stopping condition is reached. For these problems, there are questions involving properties such as computability and periodicity that do not arise when solutions are (finite) permutations. The results in this paper on the single-probe and unlimited-probe variants of the Unanimous Vote problem also suggest the following, more general question: what is the relationship between single-probe problems and their unlimited-probe variants?

\section*{Acknowledgements}

Thanks to Kunal Marwaha for a useful discussion on the \(d\)-ary Unanimous Vote problem. E.B. was partially supported by the NYU Tandon Undergraduate Summer Research Program, the NYU Tandon NextGenPhD Scholars program, and an NSF Graduate Research Fellowship under Award No. 2234660. This material is based upon work supported by the Air Force Office of Scientific Research under award number FA9550-23-1-0556.

\section*{LLM Usage Disclosure}

\href{https://claude.ai}{Claude} (Opus 4.6, 4.7, 5) was used to generate \MulAdGapLink{\texttt{MulAdaptivityGap.lean}} and \AddAdGapLink{\texttt{AddAdaptivityGap.lean}} using several lemmas provided by \href{https://aristotle.harmonic.fun}{Aristotle}. Claude (Opus 5) constructed the reduction in Appendix \ref{app: d-ary np-hard}, and used our ideas in \Cref{alg: greedy} to design the PTAS of \Cref{thm: d-ary ptas}. \href{https://gemini.google.com/app}{Gemini} assisted with \Cref{sec: periodicity} in the following way. We determined the periodicity of the special case $p_1 = \frac{1}{3},\ p_2 = \frac{3}{4}$ to be equivalent to the periodicity of the sequence specifying the number of powers of 2 between successive powers of 3. When asked about the periodicity of that sequence, Gemini suggested the link to mechanical words. LLMs were also used to check for mistakes and typos in proofs. LLMs did not write any part of this manuscript and the authors take full responsibility for its soundness.

\printbibliography

\appendix

\section{Missing Proofs From Section~\ref{sec: unlimited flips}}\label{app: missing section 2}

\begin{proof}[Proof of \Cref{lem: only use 1 and n}]
    Let \(\sigma_0 : \mathbb{Z}_+ \rightarrow [n]\) be a sequence of coins such that for some \(k\), \(\sigma_0(k) \neq 1\) and \(\sigma_0(k) \neq n\). Let \(k_0\) be the first such position. Consider changing the coin at position \(k_0\) (to either \(1\) or \(n\), which will be determined later) and note that \(H_{\sigma_0}(i)\) and \(T_{\sigma_0}(i)\) for \(i < k_0\) are unchanged; denote this new sequence \(\sigma_1\). We now show that \(\expect{\cost(\sigma_1)}\leq \expect{\cost(\sigma_0)}\). We can write:
    \begin{align*}
        \expect{\cost(\sigma_0)} - \expect{\cost(\sigma_1)} = \sum_{i=k_0}^\infty(H_{\sigma_0}(i) - H_{\sigma_1}(i)) + \sum_{i=k_0}^\infty(T_{\sigma_0}(i) - T_{\sigma_1}(i))
    \end{align*}
    We note that for any sequence \(\sigma\), and positive integers \(a < b \le  c\), \(H_\sigma(a, c) = H_\sigma(a, b - 1)H_{\sigma}(b, c)\) (and symmetrically for \(T_\sigma\)). We also note \(H_{\sigma_0}(a,b) = H_{\sigma_1}(a,b)\) and \(T_{\sigma_0}(a,b) = T_{\sigma_1}(a,b)\) for any \(a\leq b\) such that \(k_0\notin\llbracket a,b\rrbracket\). We can then continue:
    \begin{align*}
         &\sum_{i=k_0}^\infty(H_{\sigma_0}(i) - H_{\sigma_1}(i)) + \sum_{i=k_0}^\infty(T_{\sigma_0}(i) - T_{\sigma_1}(i))\\
         &= H_{\sigma_0}(k_0-1)\sum_{i=k_0}^\infty\left(H_{\sigma_0}(k_0,i) - H_{\sigma_1}(k_0, i)\right) + T_{\sigma_0}(k_0-1)\sum_{i=k_0}^\infty(T_{\sigma_0}(k_0,i) - T_{\sigma_1}(k_0,i))\\
         &= H_{\sigma_0}(k_0-1)\left(p_{\sigma_0(k_0)} - p_{\sigma_1(k_0)}\right)\left(1 + \sum_{i=k_0+1}^\infty H_{\sigma_0}(k_0+1,i)\right) + T_{\sigma_0}(k_0-1)\left(\bar{p}_{\sigma_0(k_0)} - \bar{p}_{\sigma_1(k_0)}\right)\left(1 +\sum_{i=k_0+1}^\infty T_{\sigma_0}(k_0+1,i)\right)\tag{\(\ast\)}
    \end{align*}
    For convenience, define
    \[
        \alpha \eqdef H_{\sigma_0}(k_0-1)\left(1 + \sum_{i=k_0+1}^\infty H_{\sigma_0}(k_0+1,i)\right),\quad\beta \eqdef T_{\sigma_0}(k_0-1)\left(1 + \sum_{i=k_0+1}^\infty T_{\sigma_0}(k_0+1,i)\right)
    \]
    and note that \(\alpha\) and \(\beta\) are constant with respect to the choice of coin at position \(k_0\). We can then continue from \((\ast)\):
    \[
        = \alpha\left(p_{\sigma_0(k_0)} - p_{\sigma_1(k_0)}\right) + \beta\left(\bar{p}_{\sigma_0(k_0)} - \bar{p}_{\sigma_1(k_0)}\right) = \alpha\left(p_{\sigma_0(k_0)} - p_{\sigma_1(k_0)}\right) + \beta\left(p_{\sigma_1(k_0)} - p_{\sigma_0(k_0)}\right) = \left(\alpha - \beta\right)\left(p_{\sigma_0(k_0)} - p_{\sigma_1(k_0)}\right)
    \]
    Recall that this expression is \(\expect{\cost(\sigma_0)} - \expect{\cost(\sigma_1)}\). If \(\alpha > \beta\), the expression is \(>0\) when \(p_{\sigma_0(k_0)} > p_{\sigma_1(k_0)}\), so we can set \(\sigma_1(k_0)\gets 1\) and have \(\expect{\cost(\sigma_1)}< \expect{\cost(\sigma_0)}\). Symmetrically, if \(\alpha < \beta\), we set \(\sigma_1(k_0) \gets n\) and can see that \(\expect{\cost(\sigma_1)} < \expect{\cost(\sigma_0)}\). If \(\alpha = \beta\), we can see that the two costs are equal regardless of \(\sigma_1(k_0)\), so we can set \(\sigma_1(k_0)\gets 1\text{ or } n\) arbitrarily.\par
    Now, let \(\ell\) be the number of positions \(i\) in \(\sigma_0\) such that \(1<\sigma_0(i)<n\). Either \(\ell\) is finite or not. If \(\ell\) is finite, we can repeat the above process \(\ell\) times in total to produce successive sequences \(\sigma_1,\sigma_2,\ldots,\sigma_\ell\) (for example, \(\sigma_2\) is defined by changing \(k_1\), the first non-extremal position in \(\sigma_1\), to either \(1\) or \(n\) such that \(\expect{\cost(\sigma_2)}\leq\expect{\cost(\sigma_1)}\)) each with one position changed from the last, and by the above we will have \(\expect{\cost(\sigma_0)}\geq\expect{\cost(\sigma_1)}\geq \expect{\cost(\sigma_2)}\geq \ldots\geq\expect{\cost(\sigma_\ell)}\).\par
    Now, if \(\ell\) is not finite, we can define the successive sequences (as above) \(\sigma_1,\sigma_2,\ldots\), producing \(\sigma_m\) for any \(m\in\mathbb{Z}_+\). Then, we can define sequence \(\sigma_\infty\) by
    \(
        \sigma_\infty(i)\eqdef \sigma_{i}(i)
    \).
    Note that we know \(\sigma_i(i) = \sigma_j(i)\) for \(1\leq i\leq j\), and that \(\sigma_\infty(i)\in\{1,n\}\) for all \(i\). We now show \(\expect{\cost(\sigma_\infty)}\leq \expect{\cost(\sigma_0)}\). Suppose for the sake of contradiction that \(\expect{\cost(\sigma_\infty)} > \expect{\cost(\sigma_0)}\), then there must exist some large enough \(M\) such that \(\sum_{i=0}^M\Pr[\cost(\sigma_\infty)>i]>\expect{\cost(\sigma_0)}\). Now, we also know
    \(
        \sum_{i=0}^M\Pr[\cost(\sigma_\infty)>i] =\sum_{i=0}^M\Pr[\cost(\sigma_M)>i] \leq \expect{\cost(\sigma_M)}
    \),
    implying \(\expect{\cost(\sigma_0)} < \expect{\cost(\sigma_M)}\), which is a contradiction since we know \(\expect{\cost(\sigma_0)}\geq \expect{\cost(\sigma_1)}\geq\ldots\geq \expect{\cost(\sigma_M)}\).
\end{proof}

\begin{proof}[Proof of \Cref{thm: greedy optimal}]
    Using \Cref{lem: only use 1 and n}, we can restrict to the case where there are two coins \(p_1 < p_n\). We consider two cases:
    \begin{enumerate}
        \item \(p_1 < p_n < \frac{1}{2}\)
        or \(\frac{1}{2}< p_1 < p_n\)
        \item
        \(p_1 \leq \frac{1}{2}\leq p_n\)
    \end{enumerate}
    \noindent
    \textbf{Case 1: \(p_1 < p_n < \frac{1}{2}\)} (the subcase \textbf{
        \(\frac{1}{2}< p_1 < p_n\)} is symmetric)\\
    We proceed by showing that any sequence \(\sigma\) where \(\sigma(k) = 1\) for some \(k\in\mathbb{Z}_+\) can be strictly improved. Consider the sequence \(\sigma'\) defined by
    \[
        \sigma'(i) = \begin{cases}
            n&i=k\\
            \sigma(i)&i\neq k
        \end{cases}
    \]
    We have, by \Cref{lem: sequence cost}
    \begin{align*}
        &\expect{\cost(\sigma)} - \expect{\cost(\sigma')}\\
        &= \sum_{i=1}^{\infty}(
            H_\sigma(i) + T_\sigma(i)
        ) - \sum_{i=1}^\infty(
            H_{\sigma'}(i)+T_{\sigma'}(i)
        )\\
        &= H(k-1)(p_1-p_n)\left(1+\sum_{i=k+1}^{\infty}H_{\sigma}(k+1,i)\right) + T(k-1)(\bar{p}_1-\bar{p}_n)\left(
            1 + \sum_{i=k+1}^{\infty}T_{\sigma}(k+1,i)
        \right)\\
        &= H(k-1)(p_1-p_n)\left(1+\sum_{i=k+1}^{\infty}H_{\sigma}(k+1,i)\right) - T(k-1)(p_1-p_n)\left(
            1 + \sum_{i=k+1}^{\infty}T_{\sigma}(k+1,i)
        \right)\\
        &= (p_1-p_n)\left(
            H(k-1)\left(1+\sum_{i=k+1}^{\infty}H_{\sigma}(k+1,i)\right) - T(k-1)\left(
            1 + \sum_{i=k+1}^{\infty}T_{\sigma}(k+1,i)
        \right)
        \right)\\
        &> 0
    \end{align*}
    where the last inequality follows from the fact that both factors are negative by the assumption that \(p_1 < p_n < \frac{1}{2}\).\par
    \noindent
    \textbf{Case 2: \(p_1 \leq \frac{1}{2}\leq p_n\)}\\
    We again proceed by showing that a sequence that violates the rule can be improved. Take \(\sigma\) to be a sequence that violates the rule, and take \(k\) to be the first position in \(\sigma\) where the rule is violated. Without loss of generality, suppose position \(k\) is \(0\)-biased and \(\sigma(k)=1\). Let \(\ell\) be maximal so that \(1=\sigma(k)=\sigma(k+1)=\ldots=\sigma(k+\ell)\). If coin \(1\) is the only coin used from position \(k\) forward, then \(\ell=\infty\).\par
    \noindent
    \textbf{Case 2.1: \(\ell < \infty\)}\\
    Consider swapping positions \(k\) and \(k+\ell+1\), and call the resulting sequence \(\sigma'\). We have
    \begin{align*}
        &\expect{\cost(\sigma)} - \expect{\cost(\sigma')}\\
        &= \sum_{i=k}^{k+\ell}(H_\sigma(i) +T_{\sigma}(i))-\sum_{i=k}^{k+\ell}(H_{\sigma'}(i)+T_{\sigma'}(i))\\
        &= H_\sigma(k-1)(p_1-p_n)\left(1+\sum_{i=k+1}^{k+\ell}H_{\sigma}(k+1,i)\right) + T_{\sigma}(k-1)(p_n - p_1)\left(
            1 + \sum_{i=k+1}^{k+\ell}T_\sigma(k+1,i)
        \right)\\
        &= H_\sigma(k-1)(p_1-p_n)\left(1+\sum_{i=k+1}^{k+\ell}H_{\sigma}(k+1,i)\right) - T_{\sigma}(k-1)(p_1 - p_n)\left(
            1+ \sum_{i=k+1}^{k+\ell}T_\sigma(k+1,i)
        \right)\\
        &=(p_1-p_n)\left(
            H_\sigma(k-1)\left(1+\sum_{i=k+1}^{k+\ell}H_{\sigma}(k+1,i)\right) - T_{\sigma}(k-1)\left(
            1+ \sum_{i=k+1}^{k+\ell}T_\sigma(k+1,i)
        \right)
        \right)\\
        &> 0
    \end{align*}
    where the final inequality follows from the fact that both factors are negative; the first factor is negative because \(p_1 < p_n\), and the second factor is negative because position \(k\) is 0-biased and by the definition of \(\ell\).\par
    \noindent
    \textbf{Case 2.2: \(\ell=\infty\)}\\
    Consider the sequence \(\sigma'\) (as in Case 1) defined by
    \[
        \sigma'(i) = \begin{cases}
            n&i=k\\
            \sigma(i)&i\neq k
        \end{cases}
    \]
    By the computation of Case 1, we have
    \begin{align*}
        &\expect{\cost(\sigma)} - \expect{\cost(\sigma')}\\
        &=(p_1-p_n)
        \left(
            H_\sigma(k-1)\left(1+\sum_{i=k+1}^{\infty}H_{\sigma}(k+1,i)
        \right) - T_{\sigma}(k-1)
        \left(
            1 + \sum_{i=k+1}^{\infty}T_\sigma(k+1,i)
        \right)
        \right)\\
        &>0
    \end{align*}
    The final inequality again follows from \(p_1 < p_n\), that position \(k\) is 0-biased, and the definition of \(\ell\).\par
    Finally, it is also easy to verify that, for a sequence using coin \(1\) at 1-biased positions and coin \(n\) at 0-biased positions, whether it uses coin \(1\) or \(n\) at a given unbiased position has no effect on its expected cost.
\end{proof}

\section{Adaptivity Gap and Additive Adaptivity Gap}
\label{app: adgap}

\noindent
Stating our results requires formalizing the notion of an \textit{adaptive strategy} and its \textit{cost}.

\begin{definition}[Adaptive strategy]
    An adaptive strategy for the Unlimited-Flips Unanimous Vote problem is a function
    \[
        a : \{H\}^*\cup\{T\}^*\to [n]
    \]
    where \(a(\epsilon)\) specifies a choice of first coin to flip, and for \(i\geq 1\), \(a\left(H^i\right)\) (resp. \(a\left(T^i\right)\)) specifies a choice of coin to flip if \(i\) previous flips have been performed, and all showed heads (resp. tails).
\end{definition}

\begin{definition}
    Let \(a\) be an adaptive strategy for a fixed instance of the Unlimited-Flips Unanimous Vote problem. \(\cost(a)\) is defined to be a random variable equal to the number of coin flips performed until seeing both a head and a tail, if coins are flipped according to \(a\).
\end{definition}

\noindent
The following collection of adaptive strategies will be of particular interest, because on any instance, one will have minimal expected cost among all adaptive strategies (will be optimal).

\begin{definition}
    Let \(a_i\), for \(1\leq i\leq n\), be the adaptive strategy that first flips coin \(i\), then, if the first flip shows heads (resp. tails), repeatedly flips coin \(1\) (resp. coin \(n\)).
\end{definition}

\begin{fact}\label{fact: adaptive cost}
    \[
        \expect*{\cost(a_i)} = 1 + \frac{p_i}{1-p_1} + \frac{1 - p_i}{p_n}
    \]
\end{fact}


\begin{proposition}\label{prop: opt adapt}
    Let \(S_A\) be the set of all adaptive strategies. On any instance,
    \[
        a_1\in \argmin_{a\in S_A}\expect*{\cost(a)} \text{ or } a_n\in\argmin_{a\in S_A}\expect*{\cost(a)}
    \]
    That is, it is always true that at least one of \(a_1\) or \(a_n\) is optimal.
\end{proposition}
\begin{proof}
    Recalling that \(p_1<\ldots< p_n\), it is clear that there is always an optimal adaptive strategy \(a^*\) with \(a^*\left(H^i\right)=1,\ a^*\left(T^i\right)=n\) for \(i\geq 1\). The remaining variability is then the choice of the first coin to flip. The result then follows from noting that the expression of \Cref{fact: adaptive cost} is linear in \(p_i\).
\end{proof}

\noindent
Similarly to \(S_A\), we will refer to \(S_N\), the set of all non-adaptive strategies (sequences \(\mathbb{Z}_+\to\{1,n\}\)). We begin by providing a lower bound on both quantities, using an instance that is constructed essentially identically to the instance used to prove the lower bound in Claim 6.1 of \cite{dumankeles2026}.

\begin{lemma}
    The adaptivity gap of the Unlimited-Flips Unanimous Vote problem is at least \(1.2\), and the additive adaptivity gap is at least \(\frac{1}{2}\).
\end{lemma}

\begin{proof}
    Consider the instance \(p_1 = 0, p_2 = \frac{1}{2}\). We have that
    \begin{align*}
        \expect{\cost(a_2)} &= 1 + \frac{1/2}{1} + \frac{1 - 1/2}{1/2} = 2.5
    \end{align*}
    One can note that an optimal sequence will use coin \(1\) never, or exactly once. Let \(\sigma_\infty\) be the sequence never using coin \(1\) (always using coin \(2\)), and let \(\sigma_k\) for \(k \geq 1\) be the sequence using coin \(1\) only at position \(k\). We then have, by direct computation, that
    \begin{align*}
        \expect*{\cost(\sigma_\infty)} &= 1 + \left(\frac{1}{2} + \frac{1}{4} + \frac{1}{8} + \ldots\right) + \left(\frac{1}{2} + \frac{1}{4} + \frac{1}{8} + \ldots\right) = 3
    \end{align*}
    and that
    \begin{align*}
        \expect*{\cost(\sigma_k)} &= 1 + \sum_{j=1}^{k-1}p_2^j + \left(\sum_{j=1}^{k-1}\bar{p}_2^j + \bar{p}_2^{k-1} + \sum_{j=k}^\infty \bar{p}_2^{j} \right)\\
        &= 1 + \sum_{j=1}^{k-1}\frac{1}{2^j} + \left(\sum_{j=1}^{k-1}\frac{1}{2^j} + \frac{1}{2^{k-1}} + \sum_{j=k}^\infty\frac{1}{2^j}\right)\\
        &= 1 + \left(\sum_{j=1}^{k-1}\frac{1}{2^j} + \frac{1}{2^{k-1}}\right) + \left(\sum_{j=1}^{k-1}\frac{1}{2^j} + \sum_{j=k}^\infty\frac{1}{2^j}\right)\\
        &= 3
    \end{align*}
    We can then observe that the expected costs of \(\sigma_\infty,\ \sigma_k\) are identical, and independent of \(k\). Then, on this instance, we have that
    \[
        \frac{\inf_{\sigma\in S_N}\expect*{\cost(\sigma)}}{\inf_{a\in S_A}\expect*{\cost(a)}}\geq \frac{\inf_{\sigma\in S_N}\expect*{\cost(\sigma)}}{\expect*{\cost(a_2)}} = \frac{3}{2.5} = 1.2
    \]
    and that
    \[
        \inf_{\sigma\in S_N}\expect*{\cost(\sigma)} - \inf_{a\in S_A}\expect*{\cost(a)} \geq \inf_{\sigma\in S_N}\expect*{\cost(\sigma)} - \expect*{\cost(a_2)} = 3 - 2.5 = 1/2
    \]
    This completes the proof.
\end{proof}

\noindent
We now present a collection of sequences that will be used to prove both upper bounds. The general strategy of the proof is to argue that, on any instance, at least one among the following collection of sequences has an expected cost within the desired upper bound.


\begin{lemma}\label{lem: seq1}
    Let \(\sigma_1\) be the sequence using only coin \(n\). We have
    \[
        \expect*{\cost(\sigma_1)} = \frac{1}{p_n} + \frac{p_n}{1-p_n}
    \]
\end{lemma}

\begin{proof}
    Applying \Cref{lem: sequence cost}:
    \begin{align*}
        \expect*{\cost(\sigma_1)} &= 1 + \sum_{k=1}^\infty p_n^k + \sum_{k=1}^\infty\bar{p}_n^k = 1 +  \frac{p_n}{1-p_n} + \frac{1-p_n}{p_n} = \left(1 + \frac{1 - p_n}{p_n}\right) + \frac{p_n}{1 - p_n} = \frac{1}{p_n} + \frac{p_n}{1-p_n}
    \end{align*}
\end{proof}


\begin{lemma}\label{lem: seq2}
    Let \(\sigma_2\) be the sequence using coin \(1\) first, then coin \(n\) from the second position onwards. We have 
    \[
        \expect*{\cost(\sigma_2)} = 1 + \frac{p_1}{1-p_n} + \frac{1-p_1}{p_n}
    \]
\end{lemma}

\begin{proof}
    Applying \Cref{lem: sequence cost}:
    \begin{align*}
        \expect*{\cost(\sigma_2)} &= 1 + p_1\sum_{k=0}^\infty p_n^k + \bar{p}_1\sum_{k=0}^\infty \bar{p}_n^k = 1 + \frac{p_1}{1-p_n} + \frac{1-p_1}{p_n}
    \end{align*}
\end{proof}


\begin{lemma}\label{lem: seq3}
    Let \(\sigma_3\) be the sequence that alternates as \(n,1,n,1,\ldots\). We have
    \[
        \expect*{\cost(\sigma_3)} = 1 + \frac{p_n(1 + p_1)}{1-p_1p_n} + \frac{(1-p_n)(2-p_1)}{p_1 + p_n - p_1p_n}
    \]
\end{lemma}

\begin{proof}
    Applying \Cref{lem: sequence cost}:
    \begin{align*}
        \expect*{\cost(\sigma_3)} &= 1 + \left(\highlight[pink]{p_n} + \highlight{p_np_1} + \highlight[pink]{p_n^2p_1} + \highlight{p_n^2p_1^2}+ \ldots\right)\\
        &\hspace{18pt}+\left(\highlight[green]{\bar{p}_n} + \highlight[cyan]{\bar{p}_n\bar{p}_1} + \highlight[green]{\bar{p}_n^2\bar{p}_1} + \highlight[cyan]{\bar{p}_n^2\bar{p}_1^2} + \ldots\right)\\
        &= 1 + \highlight[pink]{\frac{p_n}{1-p_np_1}} + \highlight{\frac{p_np_1}{1-p_np_1}} + \highlight[green]{\frac{\bar{p}_n}{1-\bar{p}_n\bar{p}_1}} + \highlight[cyan]{\frac{\bar{p}_n\bar{p}_1}{1-\bar{p}_n\bar{p}_1}}\\
        &= 1 + \frac{p_n}{1-p_np_1} + \frac{p_np_1}{1-p_np_1} + \frac{1-p_n}{1 - (1-p_n)(1-p_1)} + \frac{(1-p_n)(1-p_1)}{1- (1-p_n)(1-p_1)}\\
        &= 1 + \frac{p_n(1 + p_1)}{1-p_1p_n} + \frac{(1-p_n)(2-p_1)}{p_1 + p_n - p_1p_n}
    \end{align*}
\end{proof}

\begin{lemma}\label{lem: seq4}
    Let \(\sigma_4\) be the sequence \(1,n,1,n,n,n,\ldots\) where the dots indicate that the choice of \(n\) continues. We have
    \[
        \expect*{\cost(\sigma_4)} = 1 + p_1\left(1 + p_n  + \frac{p_1p_n}{1-p_n}\right) + (1-p_1)\left(2 - p_n +\frac{(1-p_1)(1-p_n)}{p_n}\right)
    \]
\end{lemma}

\begin{proof}
    Applying \Cref{lem: sequence cost}:
    \begin{align*}
        \expect*{\cost(\sigma_4)} &= 1 + (p_1 + p_1p_n + p_1^2p_n  + p_1^2p_n^2 + p_1^2p_n^3 + \ldots)\\
        &\hspace{18pt}+ (\bar{p}_1 + \bar{p}_1\bar{p}_n + \bar{p}^2_1\bar{p}_n + \bar{p}_1^2\bar{p}_n^2 + \bar{p}_1^2\bar{p}_n^3 + \ldots)\\
        &= 1 + \left(p_1 + p_1p_n  + \frac{p_1^2p_n}{1-p_n}\right) + \left(\bar{p}_1 + \bar{p}_1\bar{p}_n + \frac{\bar{p}_1^2\bar{p}_n}{1-\bar{p}_n}\right)\\
        &= 1 + p_1\left(1 + p_n  + \frac{p_1p_n}{1-p_n}\right) + \bar{p}_1\left(1 + \bar{p}_n +\frac{\bar{p}_1\bar{p}_n}{1-\bar{p}_n}\right)\\
        &= 1 + p_1\left(1 + p_n  + \frac{p_1p_n}{1-p_n}\right) + (1-p_1)\left(2 - p_n +\frac{(1-p_1)(1-p_n)}{p_n}\right)
    \end{align*}
\end{proof}

\begin{lemma}\label{lem: adapt cost 1 n}
    \begin{align*}
        \expect*{\cost(a_1)} &= \frac{1}{1-p_1} + \frac{1-p_1}{p_n}\\
        \expect*{\cost(a_n)} &= \frac{p_n}{1 - p_1} + \frac{1}{p_n}
    \end{align*}
\end{lemma}

\noindent
\Cref{lem: adapt cost 1 n} follows easily from \Cref{fact: adaptive cost}. In the proof of \Cref{thm: adgap,thm: add adgap}, we will assume without loss of generality that \(p_1 \leq \bar{p}_n\) (equivalently, that \(p_1 + p_n \leq 1\)). We first show that under this condition, \(a_n\) is an optimal adaptive strategy.

\begin{lemma}\label{lem: a_n better}
    If \(p_1 \leq \bar{p}_n\), then \(\expect*{\cost(a_n)} \leq \expect*{\cost(a_1)}\).
\end{lemma}

\begin{proof}
    \begin{align*}
        \expect*{\cost(a_n)} &\leq \expect*{\cost(a_1)}
        \Leftrightarrow\frac{p_n}{1 - p_1} + \frac{1}{p_n}\leq \frac{1}{1-p_1} + \frac{1-p_1}{p_n}
        \Leftrightarrow p_n^2 + (1-p_1) \leq p_n + (1-p_1)^2\\
        \Leftrightarrow p_n^2-p_n &\leq \bar{p}_1^2-\bar{p}_1
        \Leftrightarrow -p_n(1-p_n)\leq -\bar{p}_1(1-\bar{p}_1)
        \Leftrightarrow \bar{p}_np_n\geq p_1\bar{p}_1
    \end{align*}
    where the first equivalence comes from \Cref{lem: adapt cost 1 n}. We now show \(\bar{p}_np_n\ge p_1\bar{p}_1\). Consider the function \(f(a)=a\bar{a}=a-a^2\), where \(f''(a)=-2\), so \(f\) is concave everywhere. We know \(p_1 \leq p_n \leq \bar{p}_1\) (the second inequality is from \(p_1 \leq \bar{p}_n\))  and \(f(p_1) = f(\bar{p}_1)\) which yields \(f(p_n)\geq f(p_1)\) by concavity.
\end{proof}

\begin{lemma}[\MulAdGapLink{\texttt{MulAdaptivityGap.lean}}]\label{lem: muladgap lean}
    For any \((p_1,p_n)\in\{(p_1,p_n)\in[0,1]^2\mid p_1\leq p_n\land p_1 + p_n\leq 1\}\) we have that
    \begin{align*}
        \min\left\{\frac{1}{p_n} + \frac{p_n}{1-p_n}, 1 + \frac{p_1}{1-p_n} + \frac{1-p_1}{p_n}, 1 + \frac{p_n(1 + p_1)}{1-p_1p_n} + \frac{(1-p_n)(2-p_1)}{p_1 + p_n - p_1p_n}\right\} \leq 1.2\left(\frac{p_n}{1-p_1} + \frac{1}{p_n}\right)
    \end{align*}
\end{lemma}

\begin{lemma}[\AddAdGapLink{\texttt{AddAdaptivityGap.lean}}]\label{lem: addadgap lean}
    For any \((p_1,p_n)\in\{(p_1,p_n)\in[0,1]^2\mid p_1\leq p_n\land p_1 + p_n\leq 1\}\) we have that
    \begin{gather*}
        \min\left\{\frac{1}{p_n} + \frac{p_n}{1-p_n}, 1 + \frac{p_1}{1-p_n} + \frac{1-p_1}{p_n}, 1 + p_1\left(1 + p_n + \frac{p_1p_n}{1-p_n}\right) + (1-p_1)\left(2 - p_n +\frac{(1-p_1)(1-p_n)}{p_n}\right)\right\}\\
        \centering{\leq \frac{1}{2} + \frac{p_n}{1-p_1} + \frac{1}{p_n}}
    \end{gather*}
\end{lemma}

\noindent
For a visualization of the two lemmas above, see this \DesmosLink{Desmos graph}.

\begin{proof}[Proof of \Cref{thm: adgap,thm: add adgap}]
    We first assume, without loss of generality, that \(p_n\) is closer to \(1/2\) than \(p_1\). That is, that \(p_1\leq \bar{p}_n \Leftrightarrow p_1 +p_n\leq 1\). This matches the assumption in \Cref{lem: muladgap lean,lem: addadgap lean}. Let \(\mathcal{I}\) be the set of all problem instances. We then have that
    \begin{align*}
        \sup_{I\in\mathcal{I}}\frac{\inf_{\sigma\in S_N}\expect*{\cost(\sigma)}}{\inf_{a\in S_A}\expect*{\cost(a)}} &\leq \sup_{I\in\mathcal{I}}\frac{\min_{i\in\{1,2,3\}}\left\{\expect*{\cost(\sigma_i)}\right\}}{\inf_{a\in S_A}\expect*{\cost(a)}}\\
        &= \sup_{I\in\mathcal{I}}\frac{\min_{i\in\{1,2,3\}}\left\{\expect*{\cost(\sigma_i)}\right\}}{\expect*{\cost(a_n)}}\\
        &\leq 1.2
    \end{align*}
    where the first equality uses \Cref{prop: opt adapt,lem: a_n better} and the second inequality uses \Cref{lem: seq1,lem: seq2,lem: seq3,lem: adapt cost 1 n,lem: muladgap lean}. This proves \Cref{thm: adgap}. We can similarly write
    \begin{align*}
        \sup_{I\in\mathcal{I}}\left\{\inf_{\sigma\in S_N}\expect*{\cost(\sigma)} - \inf_{a\in S_A}\expect*{\cost(a)}\right\} &\leq \sup_{I\in\mathcal{I}}\left\{\min_{i\in\{1,2,4\}}\expect{\cost(\sigma_i)} - \inf_{a\in S_A}\expect*{\cost(a)}\right\}\\
        &= \sup_{I\in\mathcal{I}}\left\{\min_{i\in\{1,2,4\}}\expect{\cost(\sigma_i)} -\expect*{\cost(a_n)}\right\}\\
        &\leq \frac{1}{2}
    \end{align*}
    where we again use \Cref{prop: opt adapt,lem: a_n better} in the first equality and \Cref{lem: seq1,lem: seq2,lem: seq4,lem: adapt cost 1 n,lem: addadgap lean} in the second inequality. This proves Theorem \ref{thm: add adgap}.
\end{proof}

\section{Missing Proofs from Section~\ref{sec: duv}}\label{app: missing duv section}

\begin{proof}[Proof of \Cref{lem: properties after t}]
    Define
    \[
        r\in \argmax_{c\in\Omega}P(z-1;c)
    \]
    We first argue that
    \[
        \frac{P(z-1;r)}{\sum_{c\in\Omega}P(z-1;c)} \geq \frac{1}{2}
    \]
    For the sake of contradiction, suppose \(\frac{P(z-1;r)}{\sum_{c\in\Omega}P(z-1;c)} < \frac{1}{2}\). We show that this necessitates that our probability of stopping is at least \(\frac{1}{2}\), regardless of the die at position \(z\). To this end, we write
    \begin{align*}
        \Pr[\cost(\pi) > z \mid \cost(\pi) > z - 1] &= \frac{\sum_{c\in\Omega}P(z-1;c)c_{\pi(z)}}{\sum_{c\in\Omega}P(z-1;c)} \leq \frac{P(z-1;r)}{\sum_{c\in\Omega}P(z-1;c)} < \frac{1}{2}
    \end{align*}
    which contradicts our assumption that \(\Pr[\cost(\pi)>z\mid\cost(\pi)>z-1]>1/2\). We now argue property (2), that is, for all \(i\in[n]\backslash\pi([z-1])\), \(r_i\geq\frac{1}{2}\). Suppose, for the sake of contradiction, that there exists some \(i'\in\pi(\{z,\ldots,n\})\) with \(r_{i'}<\frac{1}{2}\). We show that this die would have at least a \(\frac{1}{2}\) probability of stopping our rolling, conditioned on it being rolled; this contradicts the definition of \(z\). If \(i'\) is used as the \(z\)\textsuperscript{th} die in \(\pi\), we have that
    \begin{align*}
        \Pr[\cost(\pi) > z \mid \cost(\pi) > z-1] &= \frac{\sum_{c\in\Omega}P(z-1;c)c_{\pi(z)}}{\sum_{c\in\Omega}P(z-1;c)}\\
        &= \frac{P(z-1;r)r_{i'}}{\sum_{c\in\Omega}P(z-1;c)} + \frac{\sum_{c\in\Omega\backslash\{r\}}P(z-1;c)c_{i'}}{\sum_{c\in\Omega}P(z-1;c)}\\
        &\leq \frac{P(z-1;r)r_{i'}}{\sum_{c\in\Omega}P(z-1;c)} + \frac{\bar{r}_{i'}\sum_{c\in\Omega\backslash\{r\}}P(z-1;c)}{\sum_{c\in\Omega}P(z-1;c)}\\
        &= r_{i'}\left(\frac{P(z-1;r)}{\sum_{c\in\Omega}P(z-1;c)}\right) + \bar{r}_{i'}\left(1 - \frac{P(z-1;r)}{\sum_{c\in\Omega}P(z-1;c)}\right)\\
        &\leq \frac{1}{2}\left(r_{i'} + \bar{r}_{i'}\right) = \frac{1}{2}
    \end{align*}
    where the final inequality follows from the fact that since \(r_{i'}<\frac{1}{2}\), the expression is decreasing in \(\frac{P(z-1;r)}{\sum_{c\in\Omega}P(z-1;c)}\), and since \(\frac{P(z-1;r)}{\sum_{c\in\Omega}P(z-1;c)} \geq \frac{1}{2}\), the expression is maximized at \(\frac{P(z-1;r)}{\sum_{c\in\Omega}P(z-1;c)} = \frac{1}{2}\). This completes the contradiction.
\end{proof}

\begin{proof}[Proof of \Cref{thm: greedy less than opt plus one}]
    Let \(\pi\) be the permutation produced by \Cref{alg: greedy}. We can write
    \begin{align*}
        \expect{\cost(\pi)} = \sum_{i=0}^{n-1}\Pr[\cost(\pi) > i] = 2 + \sum_{i=2}^{n-1}\Pr[\cost(\pi) > i]
    \end{align*}
    We first handle the case corresponding to step 1 of \Cref{alg: greedy}. That is, there is a color \(b\) such that \(b_i\geq 1/2\) for all \(i\in[n]\). Recall that in this case, \(\pi\) is sorted such that \(b_{\pi(1)}\leq\ldots\leq b_{\pi(n)}\). Let \(\OPT_B\) be the expected cost of an optimal adaptive strategy that rolls at least two dice and stops rolling when it observes some \(X_i\neq b\); clearly, \(\OPT_B\leq\OPT\). We can write
    \begin{align*}
        \expect{\cost(\pi)} &= 2+\sum_{i=2}^{n-1}\Pr[\cost(\pi)>i] = 2 +\sum_{i=2}^{n-1}\sum_{c\in\Omega}P(i;c) = 2+\sum_{i=2}^{n-1}\prod_{j=1}^ib_{\pi(j)} + \sum_{i=2}^{n-1}\sum_{c\in\Omega\backslash\{b\}}P(i;c)\\
        &\leq 2+\sum_{i=2}^{n-1}\prod_{j=1}^ib_{\pi(j)} + \sum_{i=2}^{n-1}\prod_{j=1}^i\bar{b}_{\pi(j)} \leq \OPT_B + \sum_{i=2}^{n-1}\frac{1}{2^i} \leq \OPT + \frac{1}{2}
    \end{align*}
    which proves the theorem in this case. We now assume that step 1 does not terminate the algorithm.\par
    Now, consider if \(\pi\) has no crossover point; we can see that \(\expect{\cost(\pi)} \leq 2 + \frac{1}{2} + \frac{1}{4} + \ldots \leq 3\), and noting that \(\OPT\geq 2\) we are done. So, assume \(\pi\) has a crossover point \(t\). Bounding the terms before the crossover point, we can write
    \begin{align*}
        \expect*{\cost(\pi)} &\leq 2 + \frac{1}{2} + \frac{1}{4} + \dots + \frac{1}{2^{t-2}} + \sum_{i=t}^{n-1}\Pr[\cost(\pi) > i]\\
        &= 2 + \frac{1}{2} + \frac{1}{4} + \dots + \frac{1}{2^{t-2}} + \sum_{c\in\Omega}\sum_{i=t}^{n-1}P(i;c)
    \end{align*}
    Let \(r\) be the color guaranteed by \Cref{lem: properties after t}. We now argue that
    \[
        \sum_{c\in\Omega\backslash\{r\}}\sum_{i=t}^{n-1}P(i;c) \leq \frac{1}{2^{t-1}}
    \]
    First note that, since \(\frac{P(t-1;r)}{\sum_{c\in\Omega}P(t-1;c)} \geq \frac{1}{2}\) by \Cref{lem: properties after t}, we know that
    \begin{align*}
        \frac{\sum_{c\in\Omega\backslash\{r\}}P(t-1;c)}{\sum_{c\in\Omega}P(t-1;c)} &= 1 - \frac{P(t-1;r)}{\sum_{c\in\Omega}P(t-1;c)} \leq \frac{1}{2}\\
        \Leftrightarrow \sum_{c\in\Omega\backslash\{r\}}P(t-1;c) &\leq \frac{1}{2}\sum_{c\in\Omega}P(t-1;c) = \frac{1}{2}\Pr[\cost(\pi) > t-1] \leq \frac{1}{2^{t-1}}
    \end{align*}
    where the equality on the second line follows from the fact that we only continue rolling dice as long as the results are identical. We now proceed by induction to show that \(\sum_{c\in\Omega\backslash\{r\}}P(k;c)\leq \frac{1}{2^{k}}\) for \(k\in\llbracket t,n-1\rrbracket\). Fix any such \(k\) and suppose \(\sum_{c\in\Omega\backslash\{r\}}P(k-1;c)\leq \frac{1}{2^{k-1}}\). Since \(r_i\geq\frac{1}{2}\) for \(i\in [n]\backslash \pi([t-1])\), we know that \(r_{\pi(k)} \geq \frac{1}{2}\). We can then write
    \begin{align*}
        \sum_{c\in\Omega\backslash\{r\}}P(k;c) &= \sum_{c\in\Omega\backslash\{r\}}P(k-1;c)c_{\pi(k)}\\
        &\leq \left(\sum_{c\in\Omega\backslash\{r\}}c_{\pi(k)}\right)\left(\sum_{c\in\Omega\backslash\{r\}}P(k-1;c)\right)\\
        &= \left(
            1-r_{\pi(k)}
        \right)\left(\sum_{c\in\Omega\backslash\{r\}}P(k-1;c)\right)\\
        &\leq \frac{1}{2}\cdot\frac{1}{2^{k-1}} = \frac{1}{2^k}
    \end{align*}
    Recalling our upper bound on \(\expect*{\cost(\pi)}\), we can now write
    \begin{align*}
        \expect*{\cost(\pi)} &\leq 2 + \frac{1}{2} + \frac{1}{4} + \dots + \frac{1}{2^{t-2}} + \sum_{c\in\Omega}\sum_{i=t}^{n-1}P(i;c)\\
        &\leq 2 + \frac{1}{2} + \frac{1}{4} + \dots + \frac{1}{2^{t-2}} + \frac{1}{2^{t}} + \dots + \frac{1}{2^{n-1}} + \sum_{i=t}^{n-1}P(i;r)\\
        &\leq 3 + \sum_{i=t}^{n-1}P(i;r)
    \end{align*}
    Now consider an optimal adaptive strategy for finding some \(X_i\neq r\), such that this strategy is required to query at least two variables. Denote \(\OPT_R\) the cost of this strategy, and note that \(\OPT_R \leq \OPT\), where \(\OPT\) is the expected cost of an optimal adaptive strategy for the \(d\)-ary Unanimous Vote problem. For \(i\in[n]\), let \(r^{(k)} \eqdef \Pr[X_i = r]\) for some unique \(k\) such that \(r^{(1)}\leq \dots\leq r^{(n)}\), and define the corresponding bijection \(\sigma_{r}(i) \eqdef k\) (i.e., \(\sigma_r(i)\) is the position of die \(i\) in a permutation sorted in non-decreasing order of \(\Pr[X_i=r]\) and, for example, \(\sigma_r^{-1}(1)\) gives the index of a die with minimum \(\Pr[X_i=r]\)). We have that
    \[
        \OPT_R = 2 + \sum_{i=2}^{n-1}\prod_{k=1}^ir^{(k)}
    \]
    In \(\pi\), either the position of die \(\sigma_r^{-1}(1)\) (a die with minimum \(\Pr[X_i=r]\)) is \(< t\) or it is not. If \(\pi^{-1}(\sigma_r^{-1}(1))\geq t\) (die \(\sigma_r^{-1}(1)\) was placed at or after the crossover point) then \(\frac{1}{2}\leq r^{(1)}\leq\ldots\leq r^{(n)}\), and this contradicts the assumption that step 1 did not terminate the algorithm. So, we have \(\pi^{-1}(\sigma_r^{-1}(1)) < t\). We now argue that
    \[
        \sum_{i=t}^{n-1}P(i;r) \leq \sum_{i=2}^{n-1}\prod_{k=1}^ir^{(k)}
    \]
    Consider the first term of \(\sum_{i=t}^{n-1}P(i;r)\), which is \(P(t;r)=\prod_{k=1}^{t}r_{\pi(k)}\). Since \(\sigma_{r}^{-1}(1)\) is prior to the crossover point, \(r^{(1)}\) is included in this product. Now, \(t\) is at least \(2\), which means there are at least two factors in this product. If \(\sigma_{r}^{-1}(2)\) is similarly prior to the crossover point, then \(r^{(2)}\) is included in the product as well; if \(\sigma_{r}^{-1}(2)\) is at or after the crossover point in \(\pi\), then it will be included in the product via factor \(r_{\pi(t)}\) (since dice at or to the right of the crossover point are sorted in non-decreasing order of \(\Pr[X_i=r]\)). Thus, \(P(t;r)\leq r^{(1)}r^{(2)} = \prod_{k=1}^2r^{(k)}\), the first term of \(\sum_{i=2}^{n-1}\prod_{k=1}^ir^{(k)}\). We proceed by inductively showing that each successive term of \(\sum_{i=t}^{n-1}P(i;r)\) has each factor of a corresponding term in \(\sum_{i=2}^{n-1}\prod_{k=1}^{i}r^{(k)}\) and is thus upper bounded by it. Consider an arbitrary term (after the first) in the sum \(\sum_{i=t}^{n-1}P(i;r)\). Suppose this term is the \(m\)\textsuperscript{th} term in the sum and assume that each factor in the product \(\prod_{k=1}^{m+1}r^{(k)}\) is included in the product \(\prod_{k=1}^{t+m-1}r_{\pi(k)} = P(t+m-1;r)\) (which implies that \(\prod_{k=1}^{m+1}r^{(k)} \geq P(t+m-1;r)\)). We know that \(P(t+m;r) = P(t+m-1;r)r_{\pi(t+m)}\). The newly added factor in the product \(\prod_{k=1}^{m+2}r^{(k)}\) is \(r^{(m+2)}\); if die \(\sigma_r^{-1}(m+2)\)'s position in \(\pi\) is \(\leq t + m - 1\), we are done. If it is \(> t + m - 1\), it will be placed in position \(t + m\), since dice to the right of the crossover point are sorted in non-decreasing order of \(\Pr[X_i=r]\) and all dice with lower values of \(\Pr[X_i=r]\) appeared at positions \(\leq t + m - 1\), by our inductive assumption. Then, we have that all factors in the product \(\prod_{k=1}^{m+2}r^{(k)}\) appear in the product \(P(t+m;r)=\prod_{k=1}^{t+m}r_{\pi(k)}\), and therefore that \(\prod_{k=1}^{m+2}r^{(k)}\geq P(t+m;r)\). This completes the inductive step, which shows that \(\sum_{i=t}^{n-1}P(i;r)\leq \sum_{i=2}^{n-1}\prod_{k=1}^{i}r^{(k)}\) (since the inequality holds term-wise). Recalling our upper bound on \(\expect*{\cost(\pi)}\), we can write
    \begin{align*}
        \expect*{\cost(\pi)} - \OPT_R &\leq \left(3 + \sum_{i=t}^{n-1}P(i;r)\right) - \left(2 + \sum_{i=2}^{n-1}\prod_{k=1}^ir^{(k)}\right)\leq 1\\
        \Rightarrow \expect*{\cost(\pi)} &\leq\OPT+1
    \end{align*}
    which completes the proof.
\end{proof}

\subsection{Counterexamples to the Optimality of Greedy Rules A and B}\label{app: counterexamples}

Recall that we are considering the Unlimited-Rolls variant. For greedy rule A, suppose we have two 3-sided dice with distributions \((2/3, 1/3, 0)\) (die 1) and \((5/6 - \varepsilon, 0, 1/6 + \varepsilon)\) (die 2), for an arbitrarily small \(\varepsilon>0\), over the same colors (red, green, blue); e.g., \(\Pr[X_1=\text{red}]=2/3\) and \(\Pr[X_2 = \text{blue}]=1/6+\varepsilon\). Call the first die 1, and the second die 2. One can verify that the two unique greedy sequences are:
\[
    g_1 : 1, 2, 1, 1, \dots\text{ and } g_2 : 2, 1, 1, 1, \dots
\]
The ellipses indicate that the sequences will continue rolling die 1. Note that these two sequences will have identical expected cost; so, we select one arbitrarily, call it \(g\). We can write the expected cost of \(g\) as
\begin{align*}
    \expect*{\cost(g)} &= 1 + \left(\left(\frac{5}{6} - \varepsilon\right) + \left(\frac{5}{6} - \varepsilon\right)\left(\frac{2}{3}\right) + \left(\frac{5}{6} - \varepsilon\right)\left(\frac{2}{3}\right)^2 + \dots\right) + \left(\frac{1}{6} + \varepsilon\right)\\
    &= 1 + 1 + \left(\frac{5}{6}-\varepsilon\right)\left(\frac{2/3}{1-2/3}\right) = 2 + \left(\frac{10}{6} - 2\varepsilon\right) = \frac{22}{6} - 2\varepsilon \approx 3.66\dots
\end{align*}
Consider instead the sequence rolling only die 1, call it \(\sigma^*\). We can write
\begin{align*}
    \expect*{\cost(\sigma^*)} &= 1 + \left(\frac{2}{3} + \left(\frac{2}{3}\right)^2 + \dots\right) + \left(\frac{1}{3} + \left(\frac{1}{3}\right)^2 + \dots\right) = 1 + \frac{2/3}{1 - 2/3} + \frac{1/3}{1-1/3} = 1 + 2 + \frac{1}{2} = 3.5
\end{align*}
For greedy rule B, our counterexample is simpler. Suppose we have two 3-sided dice with distributions \((1/2, 1/3, 1/6)\) (die 1) and \((1/2, 1/6, 1/3)\) (die 2) over colors (red, green, blue). It is easy to see that any sequence will be red-biased at any position except the first, which is always unbiased on any instance. So, we can construct a sequence \(g\) following greedy rule B by only choosing die 1. One can easily compute that \(\expect{\cost(g)} = 1 + \frac{1/2}{1-1/2}+\frac{1/3}{1-1/3}+\frac{1/6}{1-1/6} = 27/10 = 2.7\). However, one can compute that the sequence that alternates between die 1 and 2 has expected number of rolls equal to \(45/17\approx2.647\).

\subsection{NP-Hardness of \textit{d}-ary Unanimous Vote}\label{app: d-ary np-hard}

We reduce from Exact Cover by 3-sets (X3C), which is \(\mathsf{NP}\)-complete~\cite{garey1979}. In X3C, we are given a set \(G = \{1,\ldots,3k\}\) for an integer \(k\), along with a family of subsets \(\{S_j\subseteq G\}\) such that \(|S_j| = 3\) for all \(j\). We are asked to determine whether or not there is a subfamily \(\{T_i\}\subseteq \{S_j\}\) such that the \(T_i\) are disjoint and \(\cup T_i = G\).

\begin{theorem}
    The \(d\)-ary Unanimous Vote problem is \(\mathsf{NP}\)-hard.
\end{theorem}
\begin{proof}
    First, we assume the number of 3-sets in the X3C instance is strictly greater than \(k\); if there are at most \(k\) subsets, the problem is trivial. Given an instance of X3C, we can construct an instance of \(d\)-ary Unanimous Vote as follows. We define the set of colors \(\Omega \eqdef \{1,\ldots,3k\}\) and dice corresponding to each subset \(S_j\), where
    \[
        \Pr[X_j = c] = \left\{
            \begin{array}{cc}
                0 &\text{if}\quad c\in S_j\\
                \frac{1}{3k-3} &\text{if}\quad c\notin S_j
            \end{array}\right\}
        = \frac{\mathbf{1}_{c\notin S_j}}{3k-3}
    \]
    Suppose that the X3C instance is a YES instance, and let \(I\) be the set of indices of the disjoint covering subfamily (that is, \(\cup_{i\in I}S_i = G\) and \(\sum_{i\in I} |S_i|=3k\)). Now consider a permutation \(\pi\) that places the dice \(I\) in its first \(k\) positions. For any permutation \(\sigma\) and position \(1 < j\leq k\), the probability that die \(\sigma(j)\) is rolled is
    \[
        \Pr[\cost(\sigma)\geq j] = \sum_{c\in\Omega}\prod_{i=1}^{j-1}\Pr[X_{\sigma(i)}=c] = \sum_{c\in\Omega}\prod_{i=1}^{j-1}\frac{\mathbf{1}_{c\notin S_{\sigma(i)}}}{3k-3} = \sum_{c\in\Omega}\frac{\mathbf{1}_{c\notin\cup_{i=1}^{j-1}S_{\sigma(i)}}}{(3k-3)^{j-1}}=\frac{3k-|\bigcup_{i=1}^{j-1}S_{\sigma(i)}|}{(3k-3)^{j-1}}\tag{\(\ast\)}
    \]
    where the final equality follows from the fact that \(|\Omega| = 3k\). In the case of \(\pi\), we have
    \[
        \frac{3k-|\bigcup_{i=1}^{j-1}S_{\pi(i)}|}{(3k-3)^{j-1}} = \frac{3k-3(j-1)}{(3k-3)^{j-1}}
    \]
    which follows from the fact that the \(S_{\pi(i)}\) are disjoint for \(1\leq i\leq k\). So, we have that
    \[
        \expect{\cost(\pi)} = \sum_{j=1}^{n}\Pr[\cost(\pi)\geq j] = 1 + \sum_{j=2}^{k}\frac{3k-3(j-1)}{(3k-3)^{j-1}}
    \]
    And, necessarily, we have
    \[
        \OPT \leq 1 + \sum_{j=2}^{k}\frac{3k-3(j-1)}{(3k-3)^{j-1}}
    \]
    Now consider if the X3C instance is a NO instance. Let \(\pi\) be any permutation, and let \(I\eqdef \cup_{i=1}^{k}\{\pi(i)\}\) be the first \(k\) elements of \(\pi\). By assumption, \(\{S_i\}_{i\in I}\) is not a cover of \(G\), since if it were, it would be disjoint since \(|I|= k\). Using the formula \((\ast)\) we have 
    \begin{align*}
        \expect{\cost(\pi)} = \sum_{j=1}^n\Pr[\cost(\pi)\geq j] = 1+\sum_{j=2}^n\frac{3k-|\bigcup_{i=1}^{j-1}S_{\pi(i)}|}{(3k-3)^{j-1}} > 1+\sum_{j=2}^k\frac{3k-|\bigcup_{i=1}^{j-1}S_{\pi(i)}|}{(3k-3)^{j-1}} \geq 1+\sum_{j=2}^k\frac{3k-3(j-1)}{(3k-3)^{j-1}}
    \end{align*}
    where the first inequality follows from the fact that \(n > k\) and the \((k+1)^{\text{th}}\) term is positive (since \(|\cup_{j=1}^{k}S_{\pi(j)}| < 3k\)). In particular, we have that
    \[
        \OPT > 1+\sum_{j=2}^k\frac{3k-3(j-1)}{(3k-3)^{j-1}}
    \]
    Since this is the same quantity that upper bounded \(\OPT\) in the YES instance, we can distinguish YES and NO instances of X3C via the optimal expected cost of the constructed \(d\)-ary Unanimous Vote instance.
\end{proof}

\subsection{PTAS for \textit{d}-ary Unanimous Vote}\label{app: d-ary ptas}

To describe and analyze the PTAS, we require the following generalized definition of the crossover point restricted to a \textit{suffix} of \(\pi\).

\begin{definition}
    Let \(\pi\) be a permutation and let \(\ell\geq2\) define the suffix \(\pi(\ell),\pi(\ell+1),\ldots,\pi(n)\). Define \(t\geq\ell\), the \textnormal{crossover point} of the suffix of \(\pi\) starting from \(\ell\), to be minimal such that, for any \(\pi'\) identical to \(\pi\) in positions \(1,\ldots,t-1\), we have
    \[
        \Pr[\cost(\pi') > t \mid \cost(\pi') > t - 1] > \frac{1}{2}
    \]
\end{definition}

\begin{algorithm}
\caption{Polynomial-Time Approximation Scheme for \(d\)-ary Unanimous Vote}\label{alg: d-ary ptas}
\vspace{2pt}
\begin{enumerate}
    \item \(K\gets\lceil1/\varepsilon\rceil\)
    \item For each possible prefix \(p\) of length \(K\):
    \begin{enumerate}
        \item Fix permutation \(\pi_p\) to have prefix \(p\).
        \item If there is \(b\in\Omega\) such that \(b_i \geq \frac{1}{2}\) for each remaining die \(i\):
        \begin{enumerate}
            \item Add remaining dice in non-decreasing order of \(b_i\) and continue to the next iteration.
        \end{enumerate}
        \item From position \(K+1\) forward, choose dice using greedy rule A until we have either reached the crossover point of the suffix of \(\pi_p\) after \(p\) or we have used all dice.
        \item If we reached the crossover point \(t\) on the previous step, let \(r\) be the color guaranteed by \Cref{lem: properties after t}. Add remaining dice in non-decreasing order of \(r_i\) from position \(t\) forward.
    \end{enumerate}
    \item Return \(
        \pi \in \argmin_p \expect*{
            \cost(\pi_p)
        }
    \)
\end{enumerate}
\vspace{-8pt}
\end{algorithm}

Let \(\pi^*\) be an optimal permutation and let \(\pi\) be the permutation produced by \Cref{alg: d-ary ptas}; let \(\pi'\) be the permutation considered by \Cref{alg: d-ary ptas} whose first \(K\) positions match \(\pi^*\). Define
\begin{align*}
    \rho^*_j &\eqdef  \Pr[\cost(\pi^*) > j] = \sum_{c\in\Omega}\prod_{i=1}^jc_{\pi^*(i)}\\
    \rho_j &\eqdef  \Pr[\cost(\pi') > j] = \sum_{c\in\Omega}\prod_{i=1}^jc_{\pi'(i)}
\end{align*}
\begin{lemma}\label{lem: opt over k}
    For \(1\leq k\leq n-1\), we have
    \[
        \frac{\expect{\cost(\pi^*)}}{k}\geq \rho_k^*
    \]
\end{lemma}
\begin{proof}
    \begin{align*}
        \expect{\cost(\pi^*)} &= 1 + \sum_{j=1}^{n-1}\rho^*_j \geq 1 + \sum_{j=1}^{k}\rho_j^* \geq 1 + k\rho_k^*\\
        \Rightarrow\frac{\expect{\cost(\pi^*)}}{k}&\geq \rho_k^*
    \end{align*}
\end{proof}

\noindent
Recall \(K\eqdef\lceil1/\varepsilon\rceil\). We first note that \(\expect{\cost(\pi)}\leq\expect{\cost(\pi')}\) and we proceed to show \(\expect{\cost(\pi')}\leq\expect{\cost(\pi^*)}(1+\varepsilon)\). Consider the case in which step 2(b) terminates the iteration corresponding to \(\pi'\). We then have
\begin{align*}
    \expect{\cost(\pi')} &= 1 + \sum_{j=1}^K\rho^*_j + \sum_{j=K+1}^{n-1}\rho_j\\
    &\leq 1 + \sum_{j=1}^K\rho^*_{j} + \sum_{j=K+1}^{n-1}\prod_{k=1}^{j}b_{\pi^*(k)} + \rho^*_{K}\sum_{j=K+1}^{n-1}\prod_{k=K+1}^j\bar{b}_{\pi'(k)}\\
    &\leq 1 + \sum_{j=1}^{n-1}\rho^*_j + \rho_{K}^*\sum_{j=K+1}^{n-1}\frac{1}{2^{j-K}}\\
    &\leq \expect{\cost(\pi^*)} + \frac{\expect{\cost(\pi^*)}}{K}\\
    &\leq \expect{\cost(\pi^*)}(1+\varepsilon)
\end{align*}
where the first inequality follows from the fact that dice in the suffix of \(\pi'\) are sorted in non-decreasing order of \(b_i\) and by the fact that \(\sum_{c\in\Omega\backslash\{b\}}\prod_{k=K+1}^jc_{\pi'(k)}\leq\prod_{k=K+1}^j\bar{b}_{\pi'(k)}\). This proves the result in this case. Now, consider the case where the suffix of \(\pi'\) strictly after \(K\) has no crossover point. We can write
\begin{align*}
    \expect{\cost(\pi')} &= 1 + \sum_{j=1}^K\rho^*_j + \sum_{j=K+1}^{n-1}\rho_j\\
    &\leq 1 + \sum_{j=1}^K\rho^*_j + \sum_{j=K+1}^{n-1}\frac{\rho_K^*}{2^{j-K}}\\
    &\leq \expect{\cost(\pi^*)} + \frac{\expect{\cost(\pi^*)}}{K}\\
    &\leq \expect{\cost(\pi^*)}(1+\varepsilon)
\end{align*}

\noindent
Our final case is that in which step 2(b) did not terminate the iteration corresponding to \(\pi'\) and the suffix of \(\pi'\) from position \(K+1\) forward has a crossover point. As in the pseudocode, \(r\in\Omega\) is the color guaranteed by \Cref{lem: properties after t}.
\begin{lemma}\label{lem: r at most rho}
    Suppose the suffix of \(\pi'\) from position \(K+1\) forward has a crossover point \(t\). For \(0\leq j \leq n - t-1\)
    \[
        \prod_{i=1}^{t+j}r_{\pi'(i)}\leq\rho_{K+j+1}^*
    \]
\end{lemma}
\begin{proof}
    Let \(s:[n - K]\to[n]\) be defined by sorting the variable indices in the last \(n - K\) positions of \(\pi^*\) in non-decreasing order of \(\Pr[X_i=r]\). E.g. \(X_{s(1)}\) has minimal probability of being color \(r\) among these variables. We have
    \begin{align*}
        \prod_{i=1}^{t+j}r_{\pi'(i)}\leq \prod_{i=1}^K{r_{\pi^*(i)}}\prod_{i=1}^{j+1}r_{s(i)}\leq \prod_{i=1}^{K+j+1}r_{\pi^*(i)}\leq \rho_{K+j+1}^*
    \end{align*}
\end{proof}

\noindent
We can now write
\begin{align*}
    \expect{\cost(\pi')} &= 1 + \sum_{j=1}^K\rho^*_j + \sum_{j=K+1}^{n-1}\rho_j\\
    &\leq 1 + \sum_{j=1}^K\rho^*_j + \sum_{j=K+1}^{t-1}\frac{\rho_K^*}{2^{j-K}} + \sum_{j=t}^{n-1}\sum_{c\in\Omega}\prod_{i=1}^jc_{\pi'(i)}\\
    &\leq 1 +\sum_{j=1}^K\rho^*_j + \sum_{j=1}^{t-K-1}\frac{\rho_K^*}{2^{j}} + \sum_{j=t}^{n-1}
    \left(
        \sum_{c\in\Omega\backslash\{r\}}\prod_{i=1}^jc_{\pi'(i)} + \rho^*_{K+j-t+1}
    \right)\\
    &\leq 1 +\sum_{j=1}^K\rho^*_j + \sum_{j=1}^{t-K-1}\frac{\rho_K^*}{2^{j}} + \sum_{j=t}^{n-1}
    \left(
        \frac{\rho_K^*}{2^{j-K}} + \rho^*_{K+j-t+1}
    \right)\\
    &= 1 +\sum_{j=1}^K\rho^*_j + \sum_{j=1}^{t-K-1}\frac{\rho_K^*}{2^{j}} + \sum_{j=t-K}^{n-K-1}
    \frac{\rho_K^*}{2^{j}} + 
    \sum_{j=t}^{n-1}\rho^*_{K+j-t+1}\\
    &\leq 1 + \sum_{j=1}^K\rho^*_j + \sum_{j=t}^{n-1}\rho^*_{K+j-t+1} + \rho_K^*\\
    &\leq \expect{\cost(\pi^*)} + \rho_K^*\\
    &\leq \expect{\cost(\pi^*)} + \frac{\expect{\cost(\pi^*)}}{K}\\
    &= \expect{\cost(\pi^*)}\left(1+\frac{1}{K}\right)\\
    &\leq \expect{\cost(\pi^*)}(1+\varepsilon)
\end{align*}
where the first inequality follows from the definition of the crossover point, the second inequality follows from \Cref{lem: r at most rho}, the third inequality follows from \Cref{lem: properties after t}, and the second-to-last inequality follows from \Cref{lem: opt over k}.

\section{Approximating the Greedy Choice with Sampling}\label{app: sampling}

In this section, we show how we can, efficiently and with high probability, estimate the greedy score of \cite{azar2011} of a particular element to within a constant factor via sampling, given that the score of the greedy choice is ``high enough''; if no element has a ``high'' greedy score, we show that this portion of the ordering has a ``low enough'' probability of being reached. Our analysis is inspired by, and nearly identical to, the analysis done in Appendix A.2 of \cite{ghuge2022}. We include our analysis in detail here for completeness. Our analysis will require having at least \(K = 32n^4\ln(2n^3)\) samples after rejection. We first show that we can either obtain this many samples after rejection with high probability or that any ordering is a good approximation.\par

Let \(\pi\) be our produced permutation. Suppose at each greedy step we randomly generate \(n^2K\) realizations (using the independent distribution of the instance of Finite Coupon Collection) and discard those for which \(\alpha\notin\mathcal{G}\). There are two cases: \(\Pr[\alpha\in\mathcal{G}]\geq\frac{1}{n}\) and \(\Pr[\alpha\in\mathcal{G}] < \frac{1}{n}\). We consider the former case first. We want that at every step we reject at most \((n^2-1)K\) samples with probability at least \(1-\frac{1}{n^2}\). Let \(R\) be the number of rejected samples (by the independence of realizations, \(R\) is a binomial random variable with mean at most \(\left(1-\frac{1}{n}\right)n^2K = n(n-1)K\)). We can write
\begin{align*}
    \Pr[R > (n^2-1)K] &= \Pr[R - n(n-1)K > (n-1)K]\\
    &\leq \exp\left(-2\cdot n^2K\cdot \left(\frac{(n-1)K}{n^2K}\right)^2\right)\\
    &= \exp\left(-2\cdot K\cdot \frac{(n-1)^2}{n^2}\right)\\
    &\leq \frac{1}{n^2}
\end{align*}
where we used Hoeffding's Inequality (see \Cref{lem: hoeffding}, which is taken from \cite{hoeffding1963}) in the first inequality and in the final inequality we used that \(e^{-x}\leq\frac{1}{x}\) for \(x > 0\). By a union bound, all greedy steps have at least \(K\) realizations after rejection with probability at least \(1-\frac{1}{n}\). So, if \(\pi\) is the permutation produced and \(\mathcal{K}\) is the event that \(K\) realizations in \(\mathcal{G}\) were obtained, we can write:
\begin{align*}
    \expect{\cost(\pi,\alpha)} &= \expect{\cost(\pi, \alpha) \innermid \mathcal{K}}\Pr[\mathcal{K}] + \expect{\cost(\pi, \alpha) \innermid \neg\mathcal{K}}\left(1 - \Pr[\mathcal{K}]\right)\\
    &\leq \expect{\cost(\pi, \alpha) \innermid \mathcal{K}}\Pr[\mathcal{K}] + n\cdot\frac{1}{n}\\
    &\leq \expect{\cost(\pi, \alpha) \innermid \mathcal{K}} + 1
\end{align*}
and we incur at most \(1\) extra roll in expectation due to the possibility of rejecting too many samples. In the other case, if \(\Pr[\alpha\in\mathcal{G}] < \frac{1}{n}\), getting a good approximation is trivial. Let \(\pi^*\) be an optimal permutation. We can write
\begin{align*}
    \expect{
        \cost(\pi^*, \alpha)
    } &= \expect{
        \cost(\pi^*, \alpha)
        \innermid
        \alpha\in\mathcal{G}
    }\Pr[
        \alpha\in\mathcal{G}
    ] + \expect{
        \cost(\pi^*, \alpha)
        \innermid
        \alpha\notin\mathcal{G}
    }
    \Pr[\alpha\notin\mathcal{G}]\\
    &\geq 0 + n\left(1 - \frac{1}{n}\right)\hspace{10pt}\text{Eliminating the first term and applying the assumption}\\
    &= n - 1\\
    &\geq \expect{\cost(\sigma, \alpha)} - 1
\end{align*}
where \(\sigma\) is any permutation.\par

In what follows, we use the marginal value notation
\(
    f_S(j)\eqdef f(S\cup\{j\})-f(S)
\).
Define
\[
    g_S(j) \eqdef \sum_{\alpha\in\mathcal{G}:f^{\alpha}(S)<1}\frac{f_S^\alpha(j)}{1-f^\alpha(S)}\Pr[\alpha\mid\alpha\in\mathcal{G}]
\]
That is, with already selected dice \(S\), \(g_S(j)\) is the greedy score of \cite{azar2011} for the instance of Submodular Ranking in our reduction (note that this is distinct from the algorithm's constructed instance of Submodular Ranking).
In what follows, we assume that we have at least \(K = 32n^4\ln(2n^3)\) samples after rejection and that \(\alpha\in\mathcal{G}\), but we drop the conditioning notation for clarity.
\begin{lemma}\label{lem:pr_sum}
    At any step, with already selected dice \(S\), we have that \(\Pr[f^\alpha(S)< 1] \leq \sum_{j\in[n]\backslash S}g_S(j)\).
\end{lemma}
\begin{proof}
\begin{align*}
    \sum_{j\in[n]\backslash S}g_S(j) &= \sum_{j\in[n]\backslash S}\expect*{\frac{f^\alpha_S(j)}{1-f^\alpha(S)}\innermid f^\alpha(S) < 1}\Pr[f^\alpha(S)<1]\\
    &\geq \expect*{\frac{f^\alpha(S\cup ([n]\backslash S)) - f^\alpha(S)}{1-f^\alpha(S)}\innermid f^\alpha(S) < 1}\Pr[f^\alpha(S)<1]\\
    &\hspace{15pt}\text{by the definition of \(f^\alpha_S(j)\) and submodularity}\\
    &=\expect*{\frac{f^\alpha([n]) - f^\alpha(S)}{1-f^\alpha(S)}\innermid f^\alpha(S) < 1}\Pr[f^\alpha(S) < 1]\\
    &=\expect{1\innermid f^\alpha(S) < 1}\Pr[f^\alpha(S)<1]\\
    &\hspace{15pt}\text{Since \(f^\alpha([n]) = 1\) for any \(\alpha\)}\\
    &=\Pr[f^\alpha(S)<1]
\end{align*}
\end{proof}

\noindent
Consider now dividing the ordering \(\pi\) into \(L'\) and \(L''\), where \(L'\) is the maximal prefix such that at any step, there exists some \(j\in[n]\backslash S\) such that \(g_S(j) \geq \frac{1}{n^2}\), and \(L''\) is the remaining suffix. Let \(C'\) and \(C''\) denote the cost incurred during \(L'\) and \(L''\), respectively. Let \(I'\) denote the set of indices in \(L'\), and let \(I''\) denote the set of indices in \(L''\). Note that \(\expect{\cost(\pi,\alpha)} = \expect{C'} + \expect{C''}\). We can bound \(\expect{C''}\) as follows:
\begin{align*}
    \expect{C''} &\leq n\Pr[\text{\(\pi\) has not terminated by the end of \(L'\)}]\\
    &= n\Pr[f^\alpha(I') < 1]\\
    &\leq n\sum_{j\in I''}g_S(j)\hspace{15pt}\text{by Lemma \ref{lem:pr_sum}}\\
    &\leq n\left(n\cdot\frac{1}{n^2}\right)\hspace{16pt}\text{by definition of \(L''\)}\\
    &= 1
\end{align*}
We will now consider \(\expect{C'}\). We use the following lemma, the proof of which is postponed.
\begin{lemma}\label{lem:prf}
    Let \(\mathcal{F}\) denote the event that for at least one step \(i\) in \(L'\), with already chosen dice \(S\), the chosen die \(j'\) has score \(g_S(j') < \frac{1}{4}\max_{j\in[n]\backslash S}g_S(j)\). We have that \(\Pr[\mathcal{F}]\leq \frac{1}{n}\).
\end{lemma}
\noindent
We can write
\begin{align*}
    \expect{C'} &= \expect{C'\innermid \neg\mathcal{F}}(1 - \Pr[\mathcal{F}]) + \expect{C'\innermid \mathcal{F}}\Pr[\mathcal{F}]\\
    &\leq \expect{C'\innermid \neg\mathcal{F}}(1 - \Pr[\mathcal{F}]) + n\cdot\frac{1}{n}\hspace{10pt}\text{by Lemma \ref{lem:prf}}\\
    &\leq \expect{C'\innermid \neg\mathcal{F}} + 1
\end{align*}
Conditioned on \(\mathcal{F}\) not occurring, each choice in \(L'\) has at least \(\frac{1}{4}\) the score of the ``correct'' greedy choice. We note that the analysis in \cite{azar2011} is resilient to this, as it only requires that, in their histogram proof, one shrinks the height of the greedy solution's histogram by a further factor of \(4\). So, we have that
\[
    \expect{\cost(\pi,\alpha)} = \expect{C'} + \expect{C''} \leq \expect{C'} + 1 \leq \expect{C'\innermid \neg \mathcal{F}} + 2 \leq \mathcal{O}(\log d)\expect{\cost(\pi^*,\alpha)}
\]
where the final inequality absorbs an extra factor of \(4\) from sampling, and follows from the approximation guarantee of \cite{azar2011} and the construction of our instance. Recall that the above was conditioned on the event \(\mathcal{K}\) that we retain at least \(K\) realizations after rejection, which means that the analysis above shows \(\expect{\cost(\pi,\alpha)\innermid\mathcal{K}} \leq \mathcal{O}(\log d)\expect{\cost(\pi^*,\alpha)}\). Since we already showed
\(\expect{\cost(\pi,\alpha)}\leq  \expect{\cost(\pi,\alpha)\innermid \mathcal{K}} + 1\), this suffices to prove the final approximation bound.\par

Before proving Lemma \ref{lem:prf}, we state a form of Hoeffding's Inequality, proven in Theorem 1 of \cite{hoeffding1963}, and a two-sided multiplicative bound that follows easily from the additive form of Hoeffding's Inequality.
\begin{lemma}[Hoeffding's Inequality]\label{lem: hoeffding}
    Let \(X_1,\dots,X_N\) be independent random variables over \([0,1]\) and define \(\bar{X} = \frac{1}{N}\sum_{i=1}^NX_i\). Let \(\mu = \expect*{\bar{X}}\). We have that, for any \(t > 0\),
    \[
        \Pr\left[\bar{X} - \mu \geq t\right] \leq \exp\left(-2Nt^2\right)
    \]
\end{lemma}

\begin{lemma}[Two-sided multiplicative form of Hoeffding's Inequality]\label{cor:sample_bound}
Let \(X_1,\dots,X_N\) be independent random variables over \([0,1]\) and define \(\bar{X} = \frac{1}{N}\sum_{i=1}^NX_i\). Let \(\mu = \expect*{\bar{X}}\). We have that, for any \(\delta > 0\),
\[
    \Pr[(1 -\delta)\mu \leq \bar{X} \leq (1 + \delta)\mu] \geq 1 - 2\exp\left(-2N\delta^2\mu^2\right)
\]
\end{lemma}

\noindent
It is worth noting that both of the above apply in our case, since the values we are estimating, \(g_S(j)\), are bounded between \(0\) and \(1\) and the realizations in our sample are independently generated. We now prove Lemma \ref{lem:prf}. Recall that at each greedy step we use \(K = 32n^4\ln(2n^3)\) realizations to estimate the scores of each remaining \(j\in[n]\backslash S\), and that we restrict our consideration to \(L'\). Fix some step in \(L'\) and partition the remaining dice into:
\begin{itemize}
    \item \(J^+ \eqdef \left\{ j^+\in[n]\backslash S : g_S(j^+) \geq \frac{1}{4n^2} \right\}\)
    \item \(J^- \eqdef \left\{ j^-\in[n]\backslash S : g_S(j^-) < \frac{1}{4n^2} \right\}\)
\end{itemize}
Consider an arbitrary \(j^+\in J^+\), and let \(\bar{g}_S(j^+)\) be its estimated score from \(K\) independent samples. Setting \(\delta = 1/2\) and applying \Cref{cor:sample_bound}, we have that
\begin{align*}
    1 - \Pr\left[\frac{1}{2}g_S(j^+) \leq \bar{g}_S(j^+) \leq 2g_S(j^+)\right] &\leq 2\exp\left(-2\cdot K\cdot \left(\frac{1}{2}\right)^2 \cdot g_S(j^+)^2\right)\\
    &\leq 2\exp\left(-2\cdot
    \left(\frac{1}{2}\right)^2\cdot
    32n^4\ln\left(2n^3\right)\cdot
    \left(\frac{1}{4n^2}\right)^2\right)\\
    &=2\exp\left(-\ln\left(2n^3\right)\right)\\
    &= \frac{1}{n^3}
\end{align*}
Now consider an arbitrary \(j^-\in J^-\). We can write
\begin{align*}
    \Pr\left[\bar{g}_S(j^-) \geq \frac{1}{2n^2}\right] &= \Pr\left[\bar{g}_S(j^-) - \frac{1}{4n^2} \geq \frac{1}{4n^2}\right]\\
    &\leq \Pr\left[\bar{g}_S(j^-) - g_S(j^-) \geq \frac{1}{4n^2}\right]\\
    &\leq \exp\left(-2 \cdot K\cdot \left(\frac{1}{4n^2}\right)^2\right)\\
    &= \exp\left(-2 \cdot 32n^4\ln\left(2n^3\right)\cdot \frac{1}{16n^4}\right)\\
    &= \exp\left(-4\ln\left(2n^3\right)\right)\\
    & \leq \frac{1}{n^3}
\end{align*}
where the first inequality is because \(g_S(j^-) < \frac{1}{4n^2}\) and the second inequality follows from Hoeffding's Inequality. From the above we can conclude that on some fixed step \(i\) and with probability at least \(1 - \frac{1}{n^3}\), a given die \(j\) has the following property:
\begin{itemize}
    \item \text{if \(j\in J^+\), \(\frac{1}{2}g_S(j)\leq \bar{g}_S(j)\leq 2g_S(j)\)}
    \item \text{if \(j\in J^-\), \(\bar{g}_S(j) < \frac{1}{2n^2}\)}
\end{itemize}
Therefore, the probability that the above holds for all dice \(j\) considered on some fixed step \(i\) is at least \(1 - \frac{1}{n^2}\). Assuming this is the case on step \(i\), consider the exact greedy choice, call it \(j^*\), and the greedy choice chosen by sampling, call it \(j'\).
\[
    g_S(j') \geq \frac{1}{2}\bar{g}_S(j') = \frac{1}{2}\max_{j\in[n]\backslash S}\bar{g}_S(j) = \frac{1}{2}\max_{j\in J^+}\bar{g}_S(j) \geq \frac{1}{4}\max_{j\in J^+}g_S(j) = \frac{1}{4}g_S(j^*)
\]
Since the above holds on a particular step \(i\) with probability at least \(1 - \frac{1}{n^2}\), it holds for all steps in \(L'\) with probability at least \(1 - \frac{1}{n}\). This suffices to prove Lemma \ref{lem:prf}.

\end{document}